%% file: TimeTravelParadoxes.tex
\documentclass[english]{article}
\usepackage{mathpazo}
\usepackage[latin9]{inputenc}
\usepackage[letterpaper]{geometry}
\usepackage{babel}
\usepackage{mathtools}
\usepackage{amsmath}
\usepackage{amsthm}
\usepackage{amssymb}
\usepackage[unicode=true,
 bookmarks=true,bookmarksnumbered=false,bookmarksopen=false,
 breaklinks=false,pdfborder={0 0 1},backref=false,colorlinks=false]
 {hyperref}

\makeatletter
\theoremstyle{plain}
\newtheorem{thm}{\protect\theoremname}[section]
\theoremstyle{definition}
\newtheorem{defn}[thm]{\protect\definitionname}
\theoremstyle{plain}
\newtheorem{prop}[thm]{\protect\propositionname}
\theoremstyle{plain}
\newtheorem{lem}[thm]{\protect\lemmaname}
\theoremstyle{plain}
\newtheorem{cor}[thm]{\protect\corollaryname}

\renewcommand\[{\begin{equation}}
\renewcommand\]{\end{equation}}

\usepackage[margin=1cm]{caption}
\usepackage{adjustbox}

\usepackage{tikz}
\usetikzlibrary{backgrounds}
\usetikzlibrary{arrows}
\usetikzlibrary{shapes,shapes.geometric,shapes.misc}

\tikzstyle{tikzfig}=[baseline=-0.25em,scale=0.5]

\pgfkeys{/tikz/tikzit fill/.initial=0}
\pgfkeys{/tikz/tikzit draw/.initial=0}
\pgfkeys{/tikz/tikzit shape/.initial=0}
\pgfkeys{/tikz/tikzit category/.initial=0}

\pgfdeclarelayer{edgelayer}
\pgfdeclarelayer{nodelayer}
\pgfsetlayers{background,edgelayer,nodelayer,main}

\tikzstyle{none}=[inner sep=0mm]

\newcommand{\tikzfig}[1]{%
{\tikzstyle{every picture}=[tikzfig]
\IfFileExists{#1.tikz}
  {\input{#1.tikz}}
  {%
    \IfFileExists{./#1.tikz}
      {\input{./#1.tikz}}
      {\tikz[baseline=-0.5em]{\node[draw=red,font=\color{red},fill=red!10!white] {\textit{#1}};}}%
  }}%
}

\tikzstyle{every loop}=[]

\pgfdeclarelayer{edgelayer}
\pgfdeclarelayer{nodelayer}
\pgfsetlayers{edgelayer,nodelayer,main}

\@ifundefined{showcaptionsetup}{}{%
 \PassOptionsToPackage{caption=false}{subfig}}
\usepackage{subfig}
\makeatother

\providecommand{\corollaryname}{Corollary}
\providecommand{\definitionname}{Definition}
\providecommand{\lemmaname}{Lemma}
\providecommand{\propositionname}{Proposition}
\providecommand{\theoremname}{Theorem}

\begin{document}
\global\long\def\Z{\mathbb{Z}}%

\global\long\def\N{\mathbb{\mathbb{N}}}%

\global\long\def\BR{\mathbb{\mathbb{R}}}%

\global\long\def\C{\mathbb{\mathbb{C}}}%

\global\long\def\M{\mathbb{\mathbb{\mathcal{M}}}}%

\global\long\def\J{\mathbb{\mathbb{\mathcal{J}}}}%

\global\long\def\H{\mathbb{\mathbb{\mathcal{H}}}}%

\global\long\def\OOO{\mathbb{\mathbb{\mathcal{O}}}}%

\input{TikZstyles.tex}
\title{Time-Travel Paradoxes and Multiple Histories}
\author{\textbf{Jacob Hauser}\textsuperscript{a,b}\thanks{\protect\href{mailto:jahc2016@mymail.pomona.edu}{jahc2016@mymail.pomona.edu},
\protect\href{https://orcid.org/0000-0002-5104-5810}{https://orcid.org/0000-0002-5104-5810}} \enskip{}and \textbf{Barak Shoshany}\textsuperscript{a,c}\thanks{\protect\href{mailto:baraksh@gmail.com}{baraksh@gmail.com}, \protect\href{http://baraksh.com/}{http://baraksh.com/},
\protect\href{https://orcid.org/0000-0003-2222-127X}{https://orcid.org/0000-0003-2222-127X}}\\
\emph{\smallskip{}
}\\
\textsuperscript{a}\emph{ Perimeter Institute for Theoretical Physics,}\\
\emph{31 Caroline Street North, Waterloo, Ontario, N2L 2Y5, Canada\medskip{}
}\\
\textsuperscript{b}\emph{ Pomona College}\\
\emph{333 North College Way, Claremont, California, 91711, USA\medskip{}
}\\
\textsuperscript{c}\emph{ Department of Physics, Brock University}\\
\emph{1812 Sir Isaac Brock Way, St. Catharines, Ontario, L2S 3A1,
Canada}}
\maketitle
\begin{abstract}
If time travel is possible, it seems to inevitably lead to paradoxes.
These include consistency paradoxes, such as the famous grandfather
paradox, and bootstrap paradoxes, where something is created out of
nothing. One proposed class of resolutions to these paradoxes allows
for multiple histories (or timelines), such that any changes to the
past occur in a new history, independent of the one where the time
traveler originated. We introduce a simple mathematical model for
a spacetime with a time machine, and suggest two possible multiple-histories
models, making use of branching spacetimes and covering spaces respectively.
We use these models to construct novel and concrete examples of multiple-histories
resolutions to time travel paradoxes, and we explore questions such
as whether one can ever come back to a previously visited history
and whether a finite or infinite number of histories is required.
Interestingly, we find that the histories may be finite and cyclic
under certain assumptions, in a way which extends the Novikov self-consistency
conjecture to multiple histories and exhibits hybrid behavior combining
the two. Investigating these cyclic histories, we rigorously determine
how many histories are needed to fully resolve time travel paradoxes
for particular laws of physics. Finally, we discuss how observers
may experimentally distinguish between multiple histories and the
Hawking and Novikov conjectures.
\end{abstract}
\tableofcontents{}

\section{Introduction}

The theory of general relativity, which describes the curvature of
spacetime and how it interacts with matter, has been verified to very
high precision over the last 100 years. As far as we can tell, general
relativity seems to be the correct theory of gravity, at least in
the regimes we can test. However, within this theory there exist certain
spacetime geometries which feature \emph{closed timelike curves (CTCs)}
or, more generally, \emph{closed causal}\footnote{Here, by ``causal'' we mean either timelike or null.}\emph{
curves (CCCs)}, thus allowing the violation of causality \cite{FTL_TT,Visser,Krasnikov,Lobo}.
The fact that these geometries are valid solutions to Einstein's equations
of general relativity indicates crucial gaps in our understanding
of gravity, spacetime, and causality.

Wormhole spacetimes and cosmological models admitting CTCs were first
explored in the decades following the discovery of general relativity
\cite{EinsteinRosen35,VanStockum38,Godel49}. Although these spacetimes
were clearly unphysical -- the wormholes were non-traversable, and
the cosmologies unrealistic -- they were followed, several decades
later, by \emph{traversable wormholes}, \emph{warp drives}, and other
spacetimes potentially supporting time travel \cite{Alcubierre94,MorrisThorne88,FrolovNovikov90,Gott91,Krasnikov98,EverettRoman97}.

These exotic geometries which allow violations of causality almost
always violate the \emph{energy conditions} \cite{Curiel:2014zba},
a set of assumptions imposed by hand and thought to ensure that matter
sources in general relativity are ``physically reasonable''. However,
it is unclear whether or not these conditions themselves are justified,
as many realistic physical models -- notably, quantum fields --
also violate some or all of the energy conditions.

In this paper, we consider two types of causality violations: \emph{consistency
paradoxes} and \emph{bootstrap paradoxes}. A familiar example of a
consistency paradox is the \emph{grandfather paradox}, where a time
traveler prevents their own birth by going to the past and killing
their grandfather before he met their grandmother. This then means
that the time traveler, having never been born, could not have gone
back in time to prevent their own birth in the first place.

More precisely, we define a consistency paradox as the absence of
a consistent evolution for appropriate initial conditions under appropriate
laws of physics. Following Krasnikov \cite{Krasnikov02}, ``appropriate
initial conditions'' are those defined on a spacelike hypersurface
in a \emph{causal region} of spacetime -- that is, a region containing
no CTCs -- and ``appropriate laws of physics'' are those which
respect locality and which allow consistent evolutions for all initial
conditions in entirely causal spacetimes.

Bootstrap paradoxes arise when certain information or objects exist
only along CTCs, and thus appear to be created from nothing. These
are classified by some as \emph{pseudo-paradoxes} because, unlike
consistency paradoxes, they do not indicate any physical contradictions
arising from reasonable assumptions \cite{Krasnikov}. Nevertheless,
they might make one feel slightly uncomfortable. Information in a
bootstrap paradox has no clear origin and does not appear to be conserved,
and events can occur which are impossible to predict from data in
a causal region of spacetime\footnote{Even in the absence of consistency paradoxes, CTCs occur in causality-violating
regions and thus behind a \emph{Cauchy horizon} \cite{hawking_ellis_1973}.
Consider a wormhole whose mouths are surrounded by vacuum and separated
more in time than in space. Then an object may, at any time, emerge
from the earlier mouth and travel to the later mouth along a CTC,
in a way unpredictable from outside the causality-violating region.}. Therefore, we explore these pseudo-paradoxes as well, identifying
the situations in which they do or do not occur in our models.

There exist several paths for addressing the potential causality violations
arising from such spacetimes \cite{Visser}. Two of these rely on
quantum effects to resolve time travel paradoxes. The \emph{Hawking
Chronology Protection Conjecture }simply suggests that \textquotedblleft the
laws of physics do not allow the appearance of {[}CTCs{]}\textquotedblright{}
\cite{Hawking92}. Under this conjecture, quantum effects or other
laws of physics ensure that the geometry of spacetime cannot be manipulated
to allow CTCs. Deutsch's quantum time travel model, also known as
D-CTCs, resolves paradoxes by modifying quantum mechanics such that
the equation of motion is no longer unitary nor linear in the presence
of CTCs \cite{Deutsch91}.

Two other approaches address causality violations without necessarily
appealing to quantum effects. The \emph{Novikov Self-Consistency Conjecture
}holds that \textquotedblleft the only solutions to the laws of physics
that can occur locally in the real universe are those which are globally
self-consistent\textquotedblright{} \cite{Novikov90}. Thus, whether
or not CTCs are physically allowed, they can never cause valid initial
conditions to evolve in a causality-violating fashion. The \emph{multiple-histories}
(or \emph{multiple-timelines}) approach encompasses models which resolve
time travel paradoxes by allowing events to occur along different
distinct histories.

In this paper, we seek to understand and resolve causality violations
classically, so the latter two approaches are of particular interest.
In the context of the Novikov conjecture, many systems which at first
glance appear to contain consistency paradoxes have in fact been shown
to support consistent solutions for all initial conditions \cite{Consortium91,Friedman97}.
Nevertheless, clear paradoxes have been formulated which are incompatible
with the Novikov conjecture. In particular, Krasnikov used a toy model
with a specific set of physical laws in a causality-violating spacetime
to develop such a paradox in \cite{Krasnikov02}. As one of very few
concrete examples of true time travel paradoxes in the literature,
Krasnikov's model is a natural environment for us to explore the multiple-histories
approach.

This exploration serves two purposes. First, the multiple-histories
approach has traditionally been presented as a \emph{branching spacetime}
model, utilizing non-Hausdorff\footnote{A topology satisfies the \emph{Hausdorff condition} (or ``is Hausdorff'')
if and only if for any two distinct points $x_{1}\ne x_{2}$ there
exist two open neighborhoods $\OOO_{1}\ni x_{1}$ and $\OOO_{2}\ni x_{2}$
such that $\OOO_{1}\cap\OOO_{2}=\emptyset$.} (or perhaps non-locally-Euclidean \cite{McCabe}) manifolds to allow
distinct futures with shared pasts \cite{Penrose79,Visser93}. However,
the actual mechanics of resolving paradoxes using a branching spacetime
has been underdeveloped in the literature, and such constructions
present considerable mathematical challenges. Therefore, by constructing
two explicit multiple-histories models -- one mimicking a branching
spacetime and the other utilizing \emph{covering spaces} -- we provide
concrete examples of the multiple-histories approach.

Second, by demonstrating that these multiple-histories models can
prevent the appearance of consistency paradoxes entirely, we show
that a Novikov-like conjecture may hold over multiple histories, reconciling
the incompatibility between Krasnikov's model and Novikov's conjecture.
In particular, this extended Novikov conjecture holds for certain
multiple-histories resolutions containing CTCs spanning a finite number
of histories. From this perspective, the traditional Novikov conjecture
is preserved when paradoxes are absent using only one history.

This paper is organized as follows. First, in Chapter \ref{sec:A-Model-for},
we describe the twisted Deutsch-Politzer time machine and Krasnikov's
paradox model. Then, in Chapter \ref{sec:Generalizing-the-Model},
we generalize this model by allowing for additional histories, additional
particles, and additional particle ``colors''.

In Chapter \ref{sec:Infinite} we describe our two models of multiple
histories, branching spacetimes and covering spaces, in more detail.
We show how they both prevent the appearance of consistency and bootstrap
paradoxes for any number of particles and colors when an unlimited
number of histories is allowed, such that every instance of time travel
leads to a new history, and a time traveler may never return to a
previous history.

In Chapter \ref{sec:Finite} we further leverage the covering space
model to determine whether a finite number of histories could be sufficient
to resolve time travel paradoxes, and if so, how many histories are
needed. We prove several useful mathematical results, and find a condition
on the number of histories required to resolve paradoxes given the
number of colors -- that the number of histories must be divisible
by the number of colors. We also show that, although consistency paradoxes
are resolved, bootstrap paradoxes still exist if the histories are
cyclic -- but they can be avoided by reinterpreting the particle
interactions in our model.

In Chapter \ref{sec:Analysis-of-Our} we analyze several aspects of
our multiple-histories models. We discuss how, even if the histories
are cyclic, an extended Novikov conjecture can still hold over a closed
causal curve connecting all of the histories together, resulting in
physical observations combining those expected from the Novikov conjecture
with those found in multiple-histories resolutions. Furthermore, we
explore how it might be possible to experimentally distinguish --
at least in principle -- between the Hawking, Novikov, branching,
and covering space scenarios.

Finally, in Chapter \ref{sec:Discussion-and-Future} we summarize
our results and suggest avenues for future exploration.

\section{\label{sec:A-Model-for}Krasnikov's Paradox Model}

\subsection{\label{subsec:Deutsch-Politzer}The Deutsch-Politzer Time Machine}

An early attempt at formalizing a consistency paradox was proposed
by Polchinski: perhaps a billiard ball traversing a wormhole time
machine might emerge in the past and collide with its past self, ensuring
that it cannot enter the wormhole in the first place. However, Echeverria,
Klinkhammer, and Thorne found an infinite set of consistent solutions
for many reasonable initial conditions, thus showing that this system
in fact possesses no paradoxes \cite{Consortium91}. Furthermore,
it has been shown that no paradoxes exist even when considering more
general physical possibilities \cite{MikheevaNovikov93}. 

A similar construction was attempted using the \emph{Deutsch-Politzer
(DP) space} \cite{Deutsch91,Politzer92} in \cite{KalyanaRama:1994ag},
and although this construction was shown to be flawed in \cite{Krasnikov97},
a modification of this construction known as the \emph{twisted Deutsch-Politzer
(TDP) space} was used in \cite{Krasnikov02,Krasnikov} to construct
a more compelling paradox.

In 1+1 spacetime dimensions with coordinates $\left(t,x\right)$,
the DP space is constructed by associating the line $(1,x)$ with
$(-1,x)$ for $-1<x<1$ in Minkowski space. The TDP space is constructed
in an analogous way, by instead associating the line $(1,x)$ with
$(-1,-x)$ for $-1<x<1$. This means that particles entering the line
at $t=1$ will emerge at $t=-1$ ``twisted'', that is, with their
spatial orientation inverted. In both cases, the associated lines
act as mouths of a wormhole. The DP and TDP spaces are illustrated
in Figures \ref{fig:DP} and \ref{fig:TDP}.

In both spacetimes, there must be singularities at $(t,x)=(\pm1,\pm1)$,
as these points cannot be included without violating the Hausdorff
condition \cite{Krasnikov}. At all other points, the spacetimes are
flat, and we can use the same coordinates we used in the original
Minkowski space, as long as we recognize that $(1,x)$ and $(-1,\pm x)$
(with plus in the case of DP and minus in the case of TDP) refer to
the same points for $-1<x<1$ \cite{Politzer92}. In this paper, we
will ignore the presence of the singularities for the sake of simplicity,
motivated by the fact that traversable wormholes in 3+1 dimensions,
for which the DP and TDP spaces are a toy model, do not in general
possess singularities.

The \emph{causality-violating region}, denoted $J^{0}\left(\M\right)$
where $\M$ is the spacetime manifold, is the set of all points $p$
which are connected to themselves by a closed causal curve. Each such
point is in its own future and past. This is depicted for the TDP
space in Figure \ref{fig:causal-structure}.
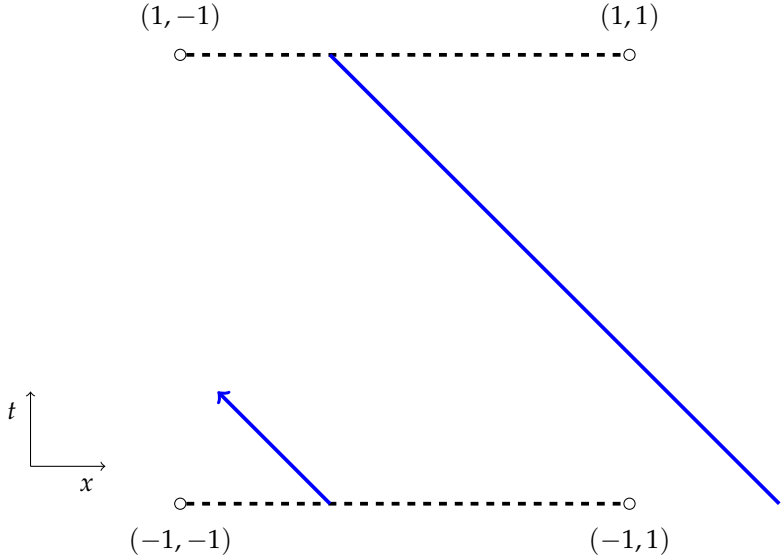
\begin{figure}[!tph]
\begin{centering}
\begin{adjustbox}{width=0.66\textwidth}\input{Figure-DP-space.tex}\end{adjustbox}
\par\end{centering}
\caption{\label{fig:DP}In the DP space, the line $(1,x)$ is associated with
$(-1,x)$ for $-1<x<1$ in Minkowski space. This is a simplified model
for a wormhole time machine \cite{Visser}. After traversing the wormhole,
the particle emerges at an earlier value of $t$ and travels in the
same direction in $x$. }
\end{figure}
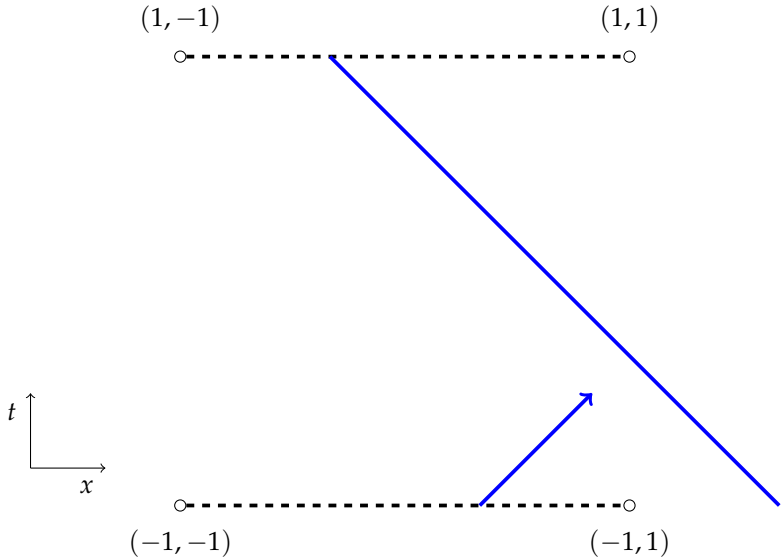
\begin{figure}[!tph]
\begin{centering}
\begin{adjustbox}{width=0.66\textwidth}\input{Figure-TDP-space.tex}\end{adjustbox}
\par\end{centering}
\caption{\label{fig:TDP}In the TDP space, $(1,x)$ is instead associated with
$(-1,-x)$ for $-1<x<1$. After emerging from the wormhole, the particle
will travel in the opposite direction in $x$.}
\end{figure}
\begin{figure}[!tph]
\centering{}\begin{adjustbox}{width=\textwidth}\input{Figure-causality-violating.tex}\end{adjustbox}\caption{\label{fig:causal-structure}The causality-violating region $J^{0}(\protect\M)$
for the TDP space $\protect\M$ is contained between the two associated
lines in $x$. The gray spacelike line indicates a choice of a reasonable
surface on which to define initial conditions. We also see particles
of two different colors, blue and green, emerging from the right and
left; the meaning of these colors is explained in Section \ref{subsec:Krasnikov's-Paradox}.}
\end{figure}
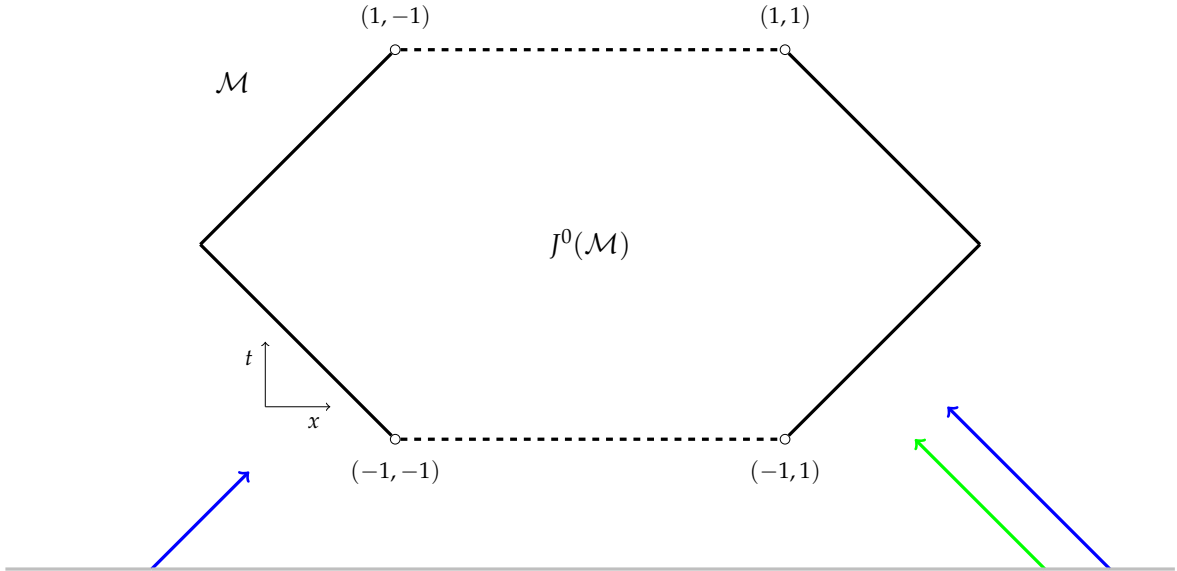

\subsection{\label{subsec:Krasnikov's-Paradox}Particles and Interaction Vertices}

Krasnikov \cite{Krasnikov02,Krasnikov} constructs a paradox in the
TDP space by introducing point particles accompanied by a set of physical
laws:
\begin{enumerate}
\item The particles are massless, and thus follow null geodesics\footnote{Note that this means we should discuss CCCs (closed causal curves,
where here ``causal'' means either timelike or null) and not CTCs
(closed timelike curves), although both types exist in the TDP space.}.
\item Whenever two particle worldlines intersect, the two particles interact.
This interaction can be interpreted as an elastic collision, with
each particle flipping its direction of movement. Later we will see
that this can lead to bootstrap paradoxes, and suggest a different
interpretation, where particles instead go through each other, continuing
in the same direction they were going.
\item Each particle has one of two colors\footnote{This property is named ``flavor'' in \cite{Krasnikov02} and ``charge''
in \cite{Krasnikov}. Here we adopt the name ``color'' in order
to make the visualization clearer, and also to avoid the impression
that this quantity is conserved.}. In every interaction, each particle flips its color (independently
of the color of the other particle), as illustrated in Figure \ref{fig:Four-vertices}.
\end{enumerate}
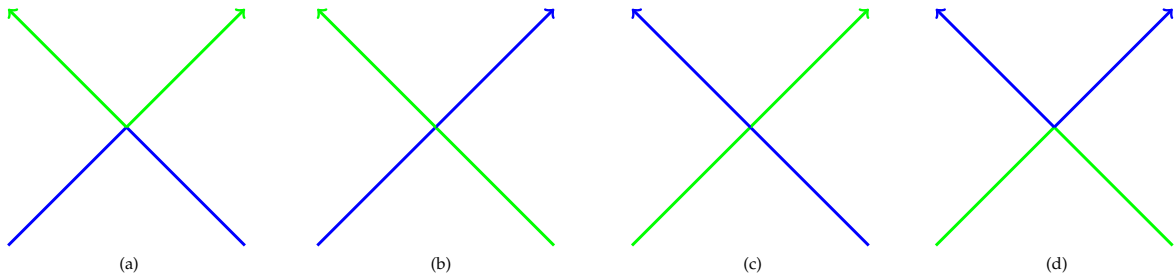
\begin{figure}[!tph]
\begin{adjustbox}{width=\textwidth}\subfloat[]{\centering{}\input{Figure-vertex1.tex}}\hfill{}\hspace{1cm}\subfloat[]{\begin{centering}
\input{Figure-vertex2.tex}
\par\end{centering}
}\hfill{}\hspace{1cm}\subfloat[]{\centering{}\input{Figure-vertex3.tex}}\hfill{}\hspace{1cm}\subfloat[]{\centering{}\input{Figure-vertex4.tex}}\end{adjustbox}\caption{\label{fig:Four-vertices}The four possible distinct vertices for
particle collisions in Krasnikov's model. Time is the vertical axis,
so the particles always come from the bottom. Note how each blue particle
changes into a green particle, and vice versa, in every collision.}
\end{figure}

The first law considerably simplifies the discussion by allowing us
to ignore timelike paths, and the second follows the spirit of Polchinski's
paradox. However, these two laws alone still permit consistent solutions
analogous to those that have been found for the Polchinski paradox,
so the third law is introduced to prevent this\footnote{Krasnikov also considers that particles appearing from the singular
points $(t,x)=(-1,\pm1)$ may allow for consistent solutions, and
introduces a fourth law to prevent this. This law adds another property
-- named ``color'' in both \cite{Krasnikov02} and \cite{Krasnikov},
but \textbf{not to be confused }with the property we call ``color''
here -- such that particles only interact with other particles of
the same ``color'', and the ``color'' itself never changes. Having
three such ``colors'' is sufficient for the purpose of preventing
consistent solutions, since there are two singularities, so they can
produce consistent solutions for at most two of the ``colors''.
In this paper, we will ignore the singularities for the sake of simplicity,
and thus the only property we will need is the one defined in law
number 3.}. These physical laws respect locality and allow consistent evolution
for all initial conditions in entirely causal spacetimes. Importantly,
one must also assume that the particles are all test particles, and
do not influence the geometry of spacetime via Einstein's equation.

We can unite all four possible vertices into a more readily generalizable
form by enumerating the two colors as $0,1\in\Z_{2}$, so that a particle's
color increases by $1\thinspace(\mathrm{mod}\thinspace2)$ after each
collision. Using these physical laws, both types of paradoxes --
consistency and bootstrap -- are illustrated in Figure \ref{fig:Paradoxes}.

\begin{figure}[!tph]
\centering{}\begin{adjustbox}{width=0.66\textwidth}\input{Figure-TDP-paradox.tex}\end{adjustbox}\caption{\label{fig:Paradoxes}An illustration of the consistency and bootstrap
paradoxes in Krasnikov's model. The blue and green lines represent
the two possible particle colors, as above. The gray lines indicate
a particle which \textbf{cannot} be assigned a consistent color.}
\end{figure}
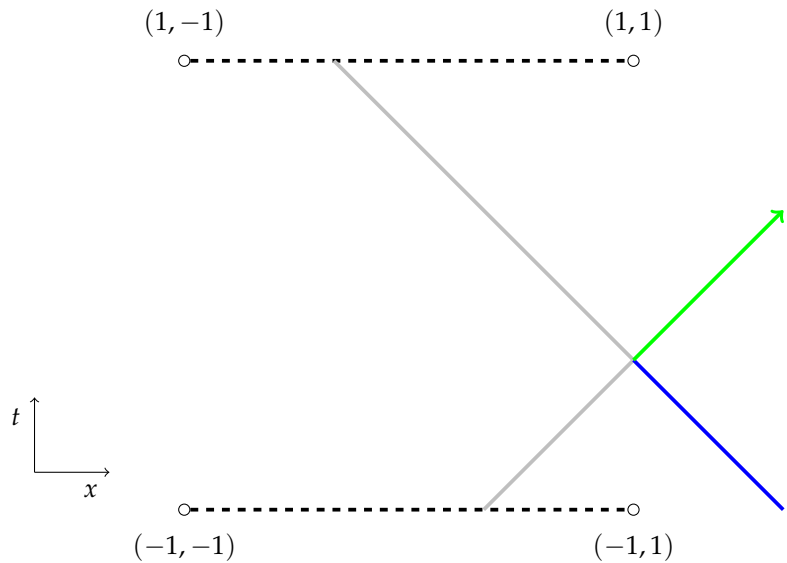

First, the particle emerging from the time machine (in gray) ends
up falling into the time machine again, so it appears out of nowhere
and exists only within the CCC -- causing a bootstrap paradox. Second,
when it collides with the particle coming from the causal region (in
blue), both must flip color. For the blue particle, this is not a
problem -- it simply changes into a green particle. However, if the
gray particle were initially blue, then it would have to change into
green, but this means it would enter the time machine as a green particle
and exit as a blue particle -- which is an inconsistency. Of course,
the same inconsistency also applies if the gray particle is initially
green. Therefore, there is no choice of color which is consistent
along the particle's entire path -- thus, we have a consistency paradox.

\section{\label{sec:Generalizing-the-Model}Generalizing the Model}

We now generalize Krasnikov's model in three ways. First, in order
to resolve the paradoxes, we introduce the possibility of multiple
histories. Next, to make sure we are considering \textbf{all }possible
initial conditions in this model, we introduce an arbitrary number
of incoming particles. Finally, to draw broader conclusions regarding
multiple-histories resolutions, we extend the model to include additional
particle colors.

\subsection{\label{subsec:Additional-Histories}Additional Histories}

In order to resolve the paradoxes demonstrated in Section \ref{subsec:Krasnikov's-Paradox},
we seek to extend our spacetime to a larger space where consistent
solutions exist. In particular, we seek to extend the TDP space to
allow for multiple, connected histories. For such an extension to
be reasonable, each history should resemble the TDP space, and all
of the histories should be identical outside the causal future of
the causality-violating region, where we expect results might differ.

With this assumption, the time traveler can go back in time to any
point in the past, and the world they will arrive at will indeed be
the \textbf{same }world from which they left, up until the moment
of arrival. However, as soon as they arrive, they inadvertently change
history -- even just by their mere presence, whether they want to
or not. Additional histories ensure that these changes can occur \textbf{independently}
of the time traveler's original history.

Initially, it may seem that only one additional history is sufficient
to resolve paradoxes -- but because each additional history should
resemble the TDP space, each introduces a new wormhole, which may
then be used to travel back in time once more. Thus, resolving paradoxes
over the entire space may require a larger number of histories --
perhaps infinitely many. We will discuss the different possibilities
in the next chapters.

Here, we consider two interesting ways to extend the TDP space. We
can depict both cases in a similar fashion, using multiple side-by-side
copies of illustrations like in Figures \ref{fig:DP}, \ref{fig:TDP},
and \ref{fig:Paradoxes}, but associating different regions of spacetime. 

First, seeking to mimic the behavior of branching spacetime models,
we can associate the line at $t=1$ in one history with the line at
$t=-1$ in the \textbf{next} history. If the events in the two histories
differ only after the wormhole mouths, then traversing the wormhole
would have the appearance of traversing a branching spacetime. An
observer would, upon traversing the time machine, appear in a new
``branch'' of the universe. The past of this branch would match
the observer's expectations, but the future could be changed without
causing an inconsistency. 

A drawback of this approach is that the first history no longer resembles
the original TDP space since it has only one wormhole mouth instead
of two -- there cannot be an \textbf{exit }to the time machine in
the first history, since there is no previous history for the time
traveler to come from. This motivates the use of covering spaces\footnote{We thank the anonymous referee for suggesting to formalize this notion
in terms of covering spaces.} as multiple-histories extensions.
\begin{defn}
Let $B$ be a topological space. A topological space $E$ is a \emph{covering
space} of $B$ if there exists a continuous surjection (or onto map)
$p:E\to B$ such that every $b\in B$ has an open neighborhood $U$
whose preimage $p^{-1}(U)$ is a union of disjoint sets $\{V_{\alpha}\}$
in $E$ where each $V_{\alpha}$ is homeomorphic to $U$ under $p$
\cite[p. 336]{munkres2000topology}.
\end{defn}

This method of extending the TDP space ensures that each history remains
faithful to the original topology of the space. In each history, there
is a wormhole entering the space and a wormhole exiting the space.
To accomplish this, the line at $t=1$ in one history is associated
with the line at $t=-1$ in the next history, as above, but the lines
at $t=-1$ in the first history and at $t=1$ in the second \textbf{also}
act as wormholes.

More rigorous constructions of specific covering spaces will follow
in Chapters \ref{sec:Infinite} and \ref{sec:Finite}. 

\subsection{Additional Particles}

We also expand the physical system under consideration by considering
an arbitrary number of particles. We will explore the general case
of $m$ particles approaching from the $-x$ direction and $n$ particles
approaching from the $+x$ direction for $m,n\in\N$, assuming that
$m\leq n$ (without loss of generality, since our spacetime is symmetric
under parity transformations $x\mapsto-x$). Considering this general
case will allow us to not only demonstrate paradoxes, but also prove
the absence of paradoxes in certain systems, motivating the utility
of the multiple-histories approach.

\subsection{Additional Colors}

In the previous chapter we expressed the two possible particle colors
as elements of $\Z_{2}$ and phrased Krasnikov's rule for color evolution
as an increase by $1\thinspace(\mathrm{mod}\thinspace2)$ to the color
of each particle after a collision. We can use the same rule for an
\textbf{arbitrary number }of colors.

Let $C\in\N$; then particle colors are elements of the cyclic group
$\Z_{C}=\left\{ 0,\ldots,C-1\right\} $ and follow the same color
evolution rule, increasing by $1\thinspace(\mathrm{mod}\thinspace C)$
after a collision. When $C=2$, we have Krasnikov's original model.
For $C>2$, we describe a more general system that will be useful
for exploring causality violations in a variety of cases. As long
as $C\neq1$, a single incoming particle leads to the same paradoxes
as in Figure \ref{fig:Paradoxes} above.

When $C>2$, particle interactions are no longer time-reversible.
However, our system has not completely lost its symmetry. In particular,
if we define \emph{color conjugation }as mapping a color $c$ to $-c\thinspace(\mathrm{mod}\thinspace C)$
in $\Z_{C}$, then CT symmetry is satisfied (with C representing color,
not charge). In fact, our interactions are symmetric under parity
transformations, so the system also satisfies CPT symmetry, as depicted
in Figure \ref{fig:Arbitrary-vertex-and}. Although the colors in
each leg of the resulting vertices will in general be different, the
system as a whole is invariant under these symmetries. 

\begin{figure}[!tph]
\centering{}\begin{adjustbox}{width=\textwidth}\subfloat[]{\centering{}\input{Figure-vertex-general.tex}}\hfill{}\hspace{2cm}\subfloat[]{\centering{}\input{Figure-vertex-cpt.tex}}\end{adjustbox}\caption{\label{fig:Arbitrary-vertex-and}(a) Given the identification between
colors and elements of $\protect\Z_{C}$, this single general vertex
captures all four vertices of Figure \ref{fig:Four-vertices} for
$C=2$, as well as those for any other values of $C$. For illustration,
the four colors in the figure -- blue, green, orange, and magenta
-- represent any of the $C$ possible colors, for the case $C\ge4$.\smallskip{}
\protect \\
(b) This vertex is the result of reversing time and parity and conjugating
color with respect to the vertex in (a). Since each particle still
leaves with a color one greater than it starts with, the result is
a valid vertex. In fact, performing CT or P transformations independently
also yields a valid vertex. In this example we took blue = 0, orange
= 1, green = 2, magenta = 3, $c=0$, $c'=2$, and $C=4$ in both (a)
and (b).}
\end{figure}
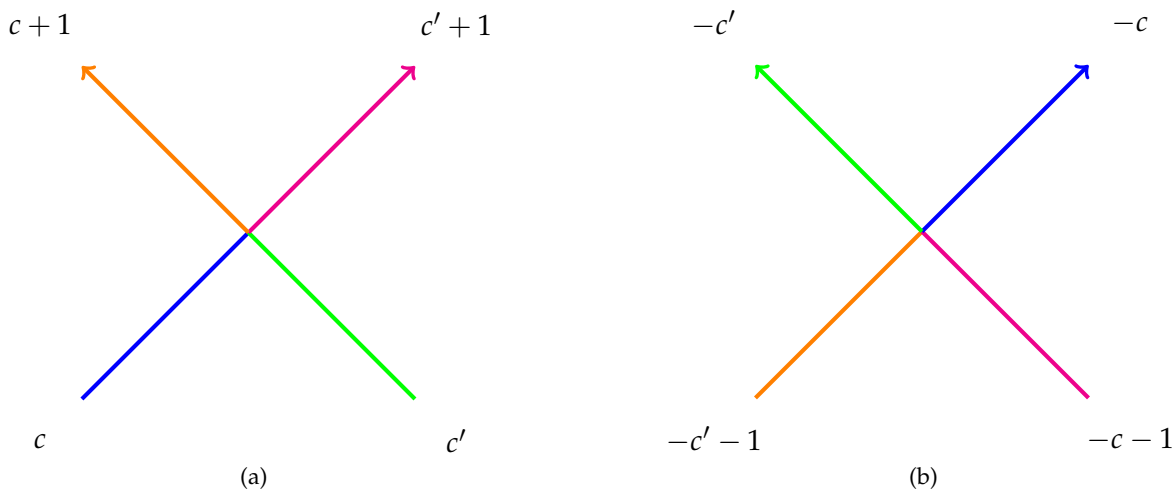

\section{\label{sec:Infinite}The Case of Unlimited Histories}

In the previous chapters we introduced a model, consisting of a particular
spacetime with specific physical laws, which admits initial conditions
for $C\geq2$ for which there is no consistent evolution, generating
a paradox. However, we also introduced the possibility for multiple
histories in some extended space. In what follows, we label these
histories with a new parameter $h\in\H$, so that points in the extended
space can be described by a triplet $(t,x,h)$. Such a parameter certainly
makes sense for branched extensions of the TDP space, where the first
space is unique and each subsequent history can be assigned a new
label. It also makes sense for a covering space extension, since the
cardinality of\footnote{Using the notation of Section \ref{subsec:Additional-Histories}.}
$\{V_{\alpha}\}=p^{-1}(U)$ is well-defined and constant, as the TDP
space is connected \cite[p. 56]{Hatcher:478079}.

In this chapter, we assume that a particle in a given history may
never return to the same history after leaving it. As before, let
$\M$ be the TDP spacetime manifold. Since the branching model has
a unique first history, it is appropriate to define $\H=\N$ and build
an extended space $\M'$ composed of a countably infinite number of
copies of Minkowski space where $(+1,x,h)$ is associated with $(-1,-x,h+1)$
for all $h\in\N$ and $-1<x<1$. This spacetime behaves as in Figure
\ref{fig:Infinite-Branching-Solution}. 

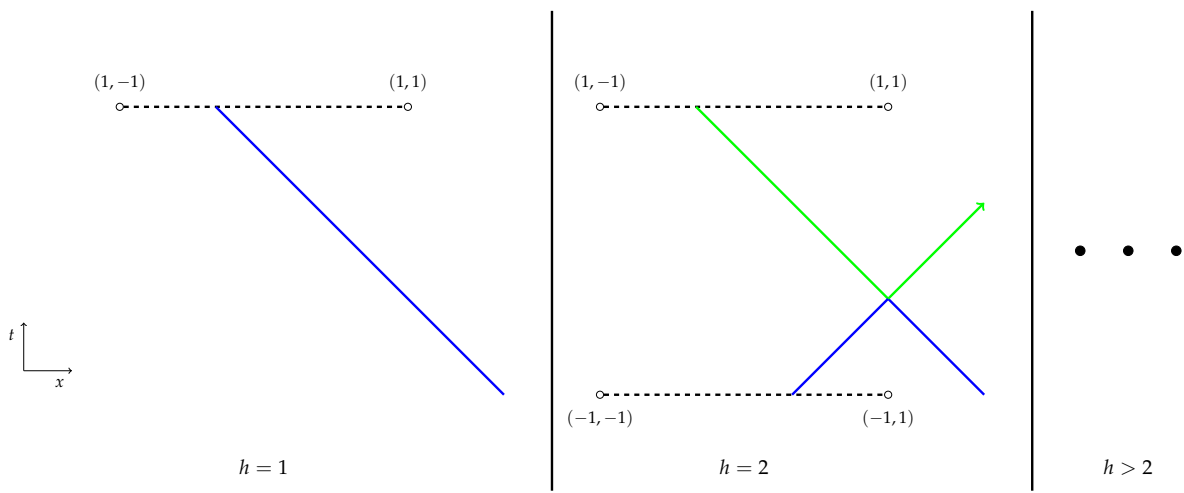
\begin{figure}[!tph]
\centering{}\begin{adjustbox}{width=\textwidth}\input{Figure-infinite-branching-solution.tex}\end{adjustbox}\caption{\label{fig:Infinite-Branching-Solution}In the branching model, when
the blue particle enters the time machine at $h=1$, it comes out
twisted (since we are in a TDP space) at $h=2$. The new history has
an identical copy of the initial blue particle, but this time it encounters
itself (or more precisely, its copy from $h=1$) and the two particles
change their colors. A green particle then enters the time machine,
and continues to $h=3$, and so on. Thus we have avoided both consistency
and bootstrap paradoxes.}
\end{figure}

In contrast, the covering space model does \textbf{not} result in
a unique first history. Consequently, it is appropriate to define
$\H=\Z$ and to build an extended space $\M'$ composed of a countably
infinite number of copies of Minkowski space where $(+1,x,h)$ is
associated with $(-1,-x,h+1)$ for all $h\in\Z$ and $-1<x<1$. This
spacetime behaves as in Figure \ref{fig:Infinite-Covering-Solution}. 

\begin{figure}[!tph]
\centering{}\begin{adjustbox}{width=\textwidth}\input{Figure-infinite-covering-solution.tex}\end{adjustbox}\caption{\label{fig:Infinite-Covering-Solution}Unlike the branching model,
the covering space model has no unique first history. Therefore, we
depict two consecutive histories $k$ and $k+1$. Without loss of
generality, a green particle emerges from the time machine in history
$k$ where it collides with the incoming blue particle; here we are
using the color convention of Figure \ref{fig:Arbitrary-vertex-and}.
Both particles increase their colors as in Figure \ref{fig:Arbitrary-vertex-and}:
blue = 0 to orange = 1 and green = 2 to magenta = 3. In history $k+1$,
the same process occurs with a magenta particle emerging from the
time machine instead of a green particle, and the magenta particle
increases its color to blue = $4\thinspace(\mathrm{mod}\thinspace4)$.
Since there is a countably infinite number of time machines, the particle
traversing the time machines never completes a CCC, nor does any copy
of the incoming blue particle. Thus, we have again avoided both consistency
and bootstrap paradoxes.}
\end{figure}
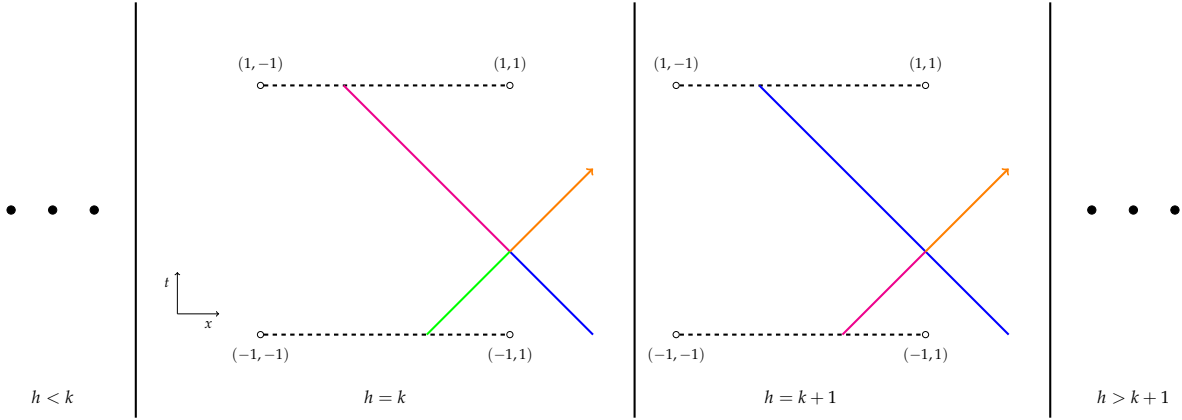

Having constructed this space, we must prove that it is indeed a covering
space of $\M$.
\begin{prop}
The extended space $\M'$ is a covering space of $\M$.
\end{prop}

\begin{proof}
Let $\M$ denote our base TDP space and let $\M'$ denote our extended
space with a countably infinite number of histories. Furthermore,
let $p:\M'\to\M$ be a covering map (or projection), defined by $(t,x,h)\xmapsto{p}(t,x)$
for $(t,x,h)\in\M'$. In order to show that $\M'$ is a covering space
of $\M$, we need to show that $p$ is both surjective (onto) and
continuous, and that each $m\in\M$ is contained in a neighborhood
$U$ whose preimage $p^{-1}(U)$ satisfies certain constraints. The
map $p$ is certainly surjective: for every $(t,x)\in\M$, each $(t,x,h)$
for the various histories $h$ is mapped to $(t,x)$ under $p$. The
map is also continuous, as it projects $\M'$ onto $\M$ without ripping
or tearing it.

For the remainder of this proof, we consider two cases: points along
the associated wormhole lines and points in the rest of the spacetime\footnote{Recall that the singularities at $(t,x)=(\pm1,\pm1)$ have been removed.}.
First, let $m=(t,x)\in\M$ be a point along an associated wormhole
line, so $t=\pm1$ and $-1<x<1$. Then, let $U$ be a ball around
$m$, small enough that it does not intersect the singularities at
$(t,x)=(\pm1,\pm1)$. Since $m$ is part of the TDP wormhole, $U$
contains points near \textbf{both} wormhole mouths, as depicted in
Figure \ref{fig:Open-Ball-1}.

The topology of $U$ is that of the union of two open balls which
intersect at a line. $p^{-1}(U)$ is composed of a countably infinite
number of such sets, now containing points from adjacent histories
$h$ and $h+1$, as in Figure \ref{fig:Open-Ball-2}. Since the number
of histories is infinite, so that no particle in a given history may
ever return to that same history, the sets in $p^{-1}(U)$ continue
in this pattern for all $h$. These sets are clearly disjoint, and
each is composed of two open balls intersecting at a line, so $p$
restricted to each is not only a bijection (since, once $p$ has been
restricted, the history data can be ignored -- rendering $p$ the
identity) but also a homeomorphism (since the topology of each $U_{k}^{+}\cup U_{k+1}^{-}$is
preserved under $p$) .

\begin{figure}[!tph]
\centering{}\begin{adjustbox}{width=0.66\textwidth}\input{Figure-open-set-1.tex}\end{adjustbox}\caption{\label{fig:Open-Ball-1}Since $m$ is a point along the associated
wormhole line, it appears twice in our representation of the TDP space
-- once at $t=-1$ and once at $t=+1$. Therefore, our ball $U$
around $m$ is actually $U=U^{+}\cup U^{-},$ the union of balls around
each wormhole mouth. It is always possible to select such a ball which
does not intersect a singularity: if $m$ is a distance $\varepsilon>0$
away from a singularity, then the ball can be chosen to have radius
$\varepsilon/2$.}
\end{figure}
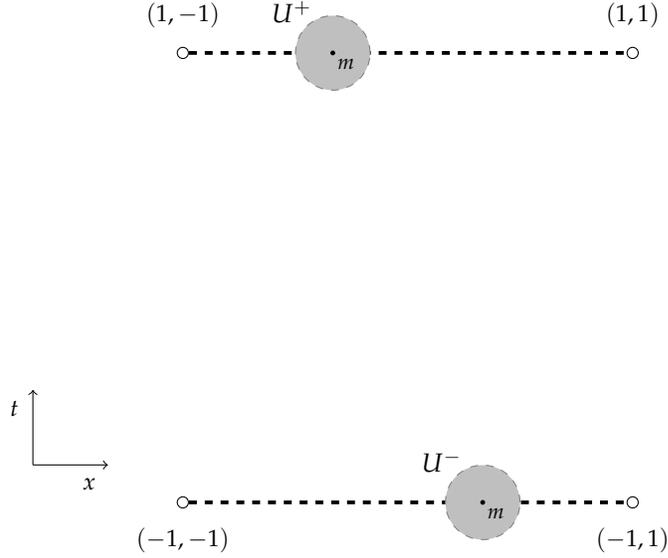

\begin{figure}[!tph]
\centering{}\begin{adjustbox}{width=\textwidth}\input{Figure-open-set-2.tex}\end{adjustbox}\caption{\label{fig:Open-Ball-2}In our extension of the TDP space, wormhole
points are now associated between adjacent histories. As a result,
the ball around the point $m_{k+1}$ -- the point overlapping histories
$k$ and $k+1$, which projects down to $m$ under the map $p$ --
is equal to $U_{k}^{+}\cup U_{k+1}^{-}.$ The preimage $p^{-1}(U)=\bigcup_{k}(U_{k}^{+}\cup U_{k+1}^{-})$
is composed of a countably infinite number of such balls, each of
which is homeomorphic to $U^{+}\cup U^{-}$ from Figure \ref{fig:Open-Ball-1}.}
\end{figure}
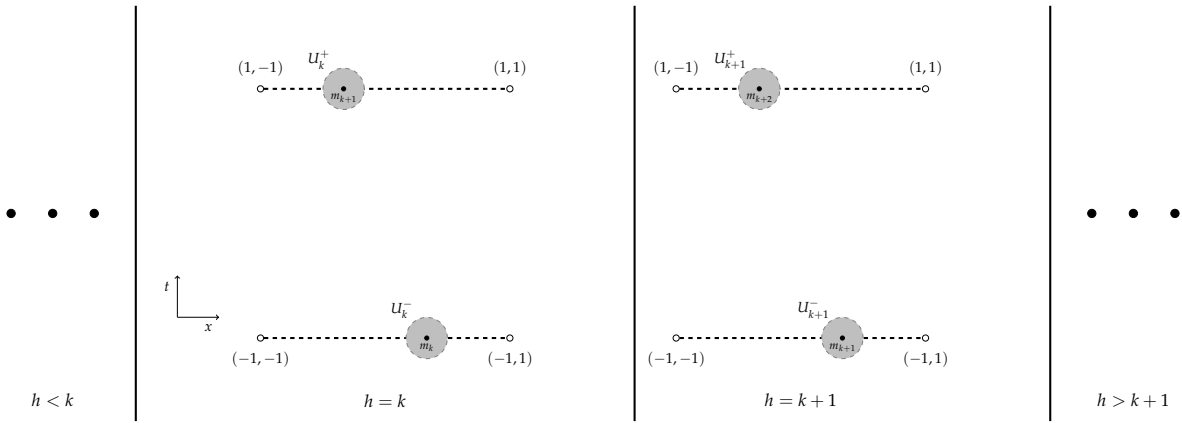

Second, let $m=(t,x)\in\M$ be a point that is not on a wormhole mouth,
and let $U$ be a ball around $m$, small enough that it intersects
neither the singularities nor the wormhole mouths. Then, $U$ has
the topology of a normal ball in flat space. Again, $p^{-1}(U)$ is
composed of a countably infinite number of such sets, which are clearly
disjoint. Also, $p$ restricted to each set acts as the identity map,
and is thus a homeomorphism. Consequently, $\M'$ is a covering space
of $\M$.
\end{proof}
This framework allows consistent solutions to our previously paradoxical
initial conditions. In both the branching model and the covering space
model, particles which would have followed CCCs in one history now
traverse multiple histories -- and since they may never return to
a previous history, they can never complete a closed loop. Consistency
paradoxes arise from conditions enforced along closed causal curves,
and bootstrap paradoxes arise from particles existing only inside
these closed curves; neither situation is possible in the unlimited
histories case, and thus both paradoxes are avoided. In other words,
we have avoided paradoxes created due to CCCs by simply avoiding any
actual CCCs\footnote{\label{fn:Although-CCCs-are}Although CCCs are no longer present,
this model still gives observers the \textbf{appearance} of time travel.
Adjacent histories are, by definition, precisely the same up until
the point in time where the time traveler exits the time machine.
Therefore, the time traveler observes a universe identical to their
initial history, but at an earlier point in time -- which is, colloquially,
the definition of time travel. However, at the moment the time traveler
exits the time machine, they have already ensured that this would
be a new history simply by existing at a point in spacetime where
they did not exist in their previous history.}.

\section{\label{sec:Finite}The Case of Finite Cyclic Histories}

Above we assumed that $h$ increases monotonically, so that the time
traveler may never return to a previous history. If there is no limit
on how many times time travel can occur -- and indeed, there is no
reason for such a limit to exist -- then this results in an infinite
number of histories. Since returning to a previous history is impossible,
and thus CCCs never form, it is straightforward to demonstrate the
absence of paradoxes.

However, we will now show that, at least within the covering space
model, it is in fact possible for a time traveler to return to a previous
history. Specifically, the covering space model provides a framework
for histories which are \textbf{cyclic} -- one can go from the last
history back to the first one. If $\M$ again denotes the TDP spacetime
manifold, then we construct an extended space $\M'$ composed of $H\in\N$
copies of Minkowski space where points are identified in the following
way:
\[
(+1,x,h)\leftrightarrow\begin{cases}
(-1,-x,h+1) & \textrm{if }1\le h\le H-1,\\
(-1,-x,1) & \textrm{if }h=H.
\end{cases}
\]

\begin{prop}
The extended space $\M'$ is a covering space of $\M$.
\end{prop}

\begin{proof}
This proof is almost identical to the one in Chapter \ref{sec:Infinite}.
There are only two differences. First, adopting the same notation,
$\{V_{\alpha}\}=p^{-1}(U)$ is composed of only $H$ disjoint sets.
Second, sets in $p^{-1}(U)$ still contain points from histories $h$
and $h+1$ when $U$ intersects a wormhole, but in this case one such
set contains points from histories $H$ and $1$. As expected, this
set is still homeomorphic to $U$ under $p$.
\end{proof}
Unlike the case of unlimited histories, this model \textbf{does }admit
CCCs, with those CCCs spanning all $H$ histories. On the other hand,
although there are only a finite number of histories, and not an infinity
of them, there nevertheless exist consistent solutions to otherwise
paradoxical scenarios, as illustrated in Figure \ref{fig:Cyclic-Resolution.}.
In fact, we will prove in Theorem \ref{thm:No-paradoxes-arise} that,
with $C$ colors and $H$ histories, paradoxes are completely avoided
in this scenario if and only if $C\mid H$.

In Section \ref{subsec:Analyzing-Particle-Collisions} we analyze
our model and determine exactly when it resolves time travel paradoxes.
In Section \ref{subsec:Avoiding-Bootstrap-Paradoxes} we discuss a
way to resolve bootstrap paradoxes, and in Section \ref{subsec:Multiple-possibilities}
we explore the general tendency of causality-violating spacetimes
to support multiple consistent solutions.

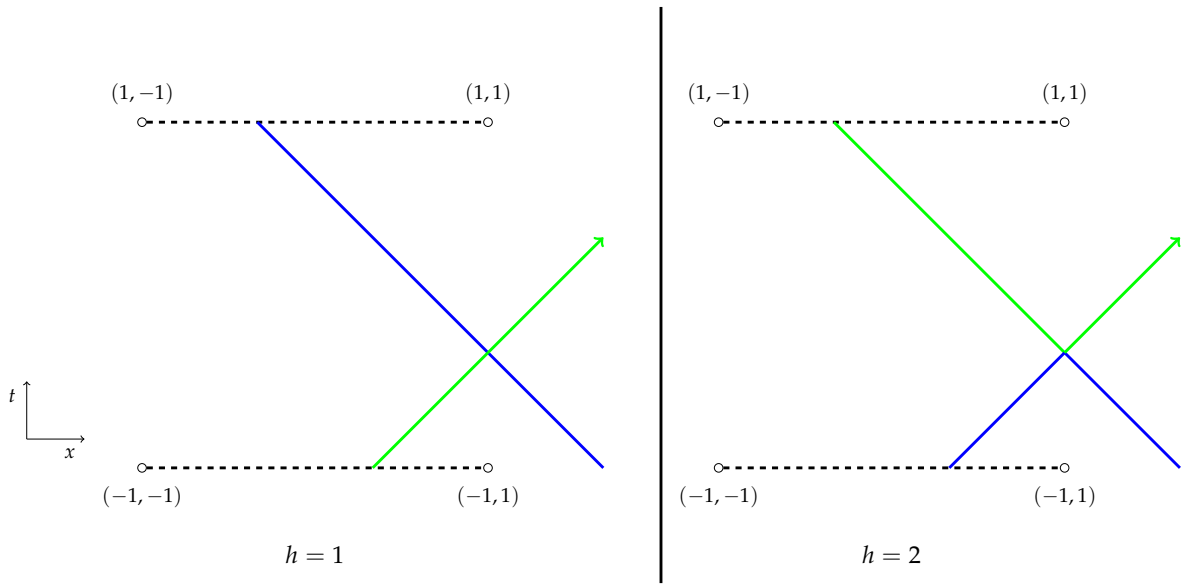
\begin{figure}[!tph]
\centering{}\begin{adjustbox}{width=\textwidth}\input{Figure-finite-solution.tex}\end{adjustbox}\caption{\label{fig:Cyclic-Resolution.}When $C=2$, the consistency paradox
can be solved with two \textbf{cyclic }histories. The blue particle
entering the time machine in $h=1$ comes out of the time machine
in $h=2$, and the green particle entering the time machine in $h=2$
comes out of the time machine back in $h=1$. Since we interpret the
vertices as elastic collisions, we now have a bootstrap paradox; the
particle traveling along the CCC only exists within the CCC itself.
We will discuss how to resolve this in Section \ref{subsec:Avoiding-Bootstrap-Paradoxes}.
Unlike in the scenario of Figure \ref{fig:Infinite-Branching-Solution},
here there is no first history where nothing has come out of the time
machine \textquotedblleft yet\textquotedblright{} (in fact, in Figure
\ref{fig:Infinite-Branching-Solution} the past exit of the time machine
does not even exist for $h=1$).}
\end{figure}

\subsection{\label{subsec:Analyzing-Particle-Collisions}How Many Histories Are
Required to Resolve Paradoxes?}

In this section, we will examine the TDP space in more detail in order
to lay the groundwork for the proof of Theorem \ref{thm:No-paradoxes-arise}.
Although initial conditions defined outside the causality-violating
region $J^{0}(\M')$ cannot uniquely determine the physics inside
this region (as will be demonstrated in Section \ref{subsec:Multiple-possibilities}),
we will show that they do determine the trajectories of all the particles
in this region\footnote{Other than those originating at singular points, as considered by
Krasnikov -- but as noted above, we will ignore this subtlety here.}. Since all particles follow null trajectories and change direction
only in elastic collisions, we can instead think of a set of particle
trajectories as straight null lines intersecting at vertices corresponding
to collisions. For example, in Figure \ref{fig:Cyclic-Resolution.},
the blue path in $h=1$ is considered to be one path whether the original
blue particle continues in the same direction or not after the interaction
at the vertex; this path then continues to $h=2$ and exits to infinity.
\begin{defn}
A \emph{particle path }is a straight null line in $\M'$ composed
of segments from the trajectories of one or more particles.
\end{defn}

Using this notion, we seek to connect particle trajectories throughout
$\M'$ to appropriate initial conditions, and to show that we need
not worry about trajectories varying across different histories or
leading to inconsistencies.
\begin{lem}
\label{lem:TDP-Null-Lines}Particle paths in all histories of $\M'$
are completely determined by initial conditions in the causal past
of the causality-violating region $J^{0}(\M')$.
\end{lem}

\begin{proof}
First, we show that all paths are extendible to $t=\pm\infty$ in
some history. This is certainly true for paths which do not enter
the time machine. As for other paths, they may traverse the wormhole
only once. Indeed, suppose that a path enters the wormhole at $\left(t,x\right)=\left(1,x_{0}\right)$
in one history and exits at $\left(t,x\right)=\left(-1,-x_{0}\right)$
in the next. Without loss of generality, we assume that the path then
moves along the $+x$ direction. Then the path, parametrized by $\lambda\in\BR$,
will be such that 
\[
\left(t,x\right)=\left(\lambda-1,\lambda-x_{0}\right).
\]
The path will reach $t=1$, where the time machine is located, at
$\lambda=2$. However, at this point it will be at $x=2-x_{0}$. The
wormhole is located at $x\in\left(-1,1\right)$, so $x_{0}$ must
be in that range, and in particular $x_{0}<1$. Hence, we see that
we must have $x>1$, and the point of intersection with $t=1$ is
outside of the wormhole. Therefore, a path can never intersect the
wormhole twice.

We conclude that all null lines entering the time machine -- including
the paths of all particles inside $J^{0}(\M')$ -- must originate
at $t=-\infty$ in some history and, upon traversing the wormhole
once, must terminate at $t=+\infty$ in another history. As a consequence,
all particles paths in $\M'$ are determined by initial data in the
causal past of the causality-violating region $J^{0}(\M')$ in some
history.
\end{proof}
\begin{cor}
\label{cor:All-histories-same}Particle paths, and the numbers of
collisions the particles on these paths experience, are the same in
all histories of $\M'$.
\end{cor}

\begin{proof}
By assumption, all histories have the \textbf{same }initial data in
the causal past of the causality-violating region $J^{0}(\M')$. In
Lemma \ref{lem:TDP-Null-Lines} we saw that all particle paths in
a history are determined by the initial conditions in that history
and the previous history. Since all such initial conditions are the
same, the particle paths in each history must also be the same. Furthermore,
since these paths fully determine the vertices denoting particle collisions,
the number of collisions the particles on these paths experience is
the same in each history.
\end{proof}
\begin{cor}
\label{cor:Consistent-Positions}The positions of all particles along
CCCs are consistent.
\end{cor}

\begin{proof}
Since all particle paths are extendible to $t=\pm\infty$ in some
history, no particles appear or disappear in a paradoxical way. Furthermore,
since only two null paths can meet at each vertex, there are no particle
interactions inconsistent with the laws of physics in this system.
\end{proof}
All that remains in order to prove the absence of paradoxes is to
demonstrate that the \textbf{colors }of particles along CCCs are consistent
as well. To do that, we first prove the following lemma:
\begin{lem}
\label{lem:Only-m-n}The color evolution of particles in $J^{0}(\M')$
is determined entirely by the choice of $m$ and $n$ -- the number
of particles entering from the left ($-x$) and the right ($+x$),
respectively.
\end{lem}

\begin{proof}
Initially, $m+n$ particles of various colors approach $J^{0}(\M')$
along various trajectories. To show that only $m$ and $n$ affect
the color evolution of particles in this region, we must show that
neither the positions of the trajectories nor the initial colors impact
this evolution.

The initial positions of the particles impact the positions of particles
in $J^{0}(\M')$, but \textbf{not} the ordering of vertices -- which
completely determine how colors change, since colors are constant
outside of the vertices. Since each null path approaching $J^{0}(\M')$
traverses the wormhole and then leaves in the same direction that
it came from (due to the twist at the wormhole), it is apparent that
$m$ particles must leave $J^{0}(\M')$ in the $-x$ direction and
$n$ particles must leave in the $+x$ direction.

Since we are considering a spacetime with one spatial dimension, the
spatial ordering of a set of particles or paths cannot change over
time, except when it is flipped passing through the wormhole. Thus,
the particles which leave $J^{0}(\M')$ must be the same particles
as those which enter it in the first place. Consequently, the remaining
particles are \textbf{confined }to $J^{0}(\M')$, and their colors
are impacted only by the number of collisions they have with the incoming
particles, not by the incoming particles' colors.
\end{proof}
Lastly, we build some machinery for analyzing the collisions of arbitrary
numbers of particles.
\begin{defn}
A \textbf{$(p,q)$ }particle\textbf{ }\emph{group collision} is a
set of individual particle collisions arising from the scattering
of $p$ particles approaching from the left ($-x$) and $q$ particles
from the right ($+x$), as in Figure \ref{fig:A-collision-of}.
\end{defn}

\begin{lem}
\label{lem:Group-collision-evolution}Let $x_{k}$ be the color of
particle number $k$ counting from the left in a $(p,q)$ particle
group collision where $p<q$. Then, after the collision, the particles'
colors are given by: 

\[
x_{k}'=x_{k}+\begin{cases}
2(k-1)+1 & k\leq p,\\
2p & p<k\leq q,\\
2(p+q-k)+1 & k>q.
\end{cases}
\]
\end{lem}

\begin{proof}
Since $p$ lines cross with $q$ lines in one of these collisions,
$p\times q$ vertices arise, as illustrated in Figure \ref{fig:A-collision-of}.
If we assign each line a number, as in the figure, then it is straightforward
to label the vertices with tuples of these numbers. Since the spatial
ordering of the particles is constant we can determine the first and
last collisions that each particle will participate in.

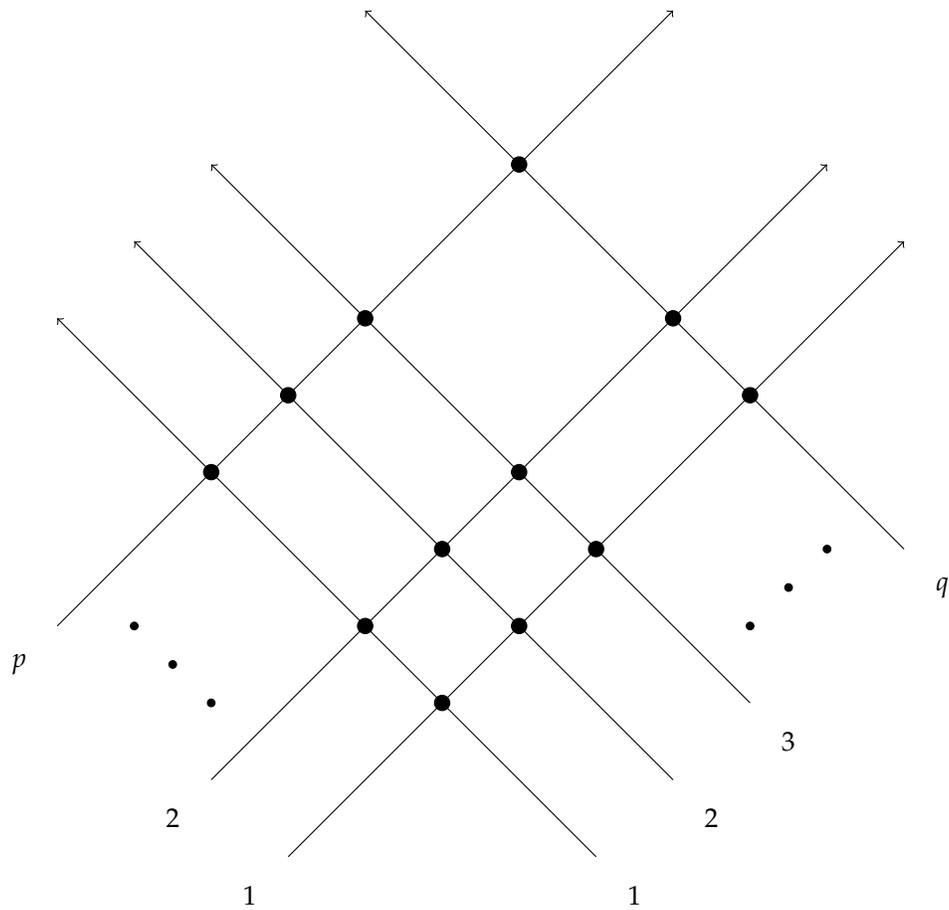
\begin{figure}[!tph]
\begin{centering}
\begin{adjustbox}{width=0.8\textwidth}\input{Figure-crossings.tex}\end{adjustbox}
\par\end{centering}
\caption{\label{fig:A-collision-of}A collision of $p$ particles from the
left and $q$ particles from the right.}
\end{figure}

In order to determine how particle colors change in one group collision,
we assign each particle non-unique initial and final vertices corresponding
to the first and last collisions that they participate in. As particles
traverse the group collision between their initial and final vertices,
they always travel from some vertex $(a,b)$ to one of two adjacent
vertices: $(a+1,b)$ or $(a,b+1)$. Thus, if a particle enters the
group collision at $(a_{i},b_{i})$ and leaves at $(a_{f},b_{f})$,
the total number of collisions that it participates in is $(a_{f}-a_{i})+(b_{f}-b_{i})+1$,
where the $+1$ accounts for the initial vertex.

Particle number $k$ first collides at $(p-k+1,1)$ if $k\leq p$
or at $(1,k-p)$ if $k>p$, and it last collides at $(p,k)$ if $k\leq q$
or at $(p+q-k+1,q)$ if $k>q$. Thus, over the course of a group collision,

\begin{equation}
x_{k}'=x_{k}+\begin{cases}
2(k-1)+1 & k\leq p,\\
2p & p<k\leq q,\\
2(p+q-k)+1 & k>q.
\end{cases}\label{eq:collision1}
\end{equation}

Note the special case where $p=q$, and 

\begin{equation}
x_{k}'=x_{k}+\begin{cases}
2(k-1)+1 & k\leq p,\\
2(2p-k)+1 & k>p.
\end{cases}\label{eq:collision2}
\end{equation}
\end{proof}
Now, having built this machinery, we proceed to the main result of
this chapter.
\begin{thm}
\label{thm:No-paradoxes-arise}No paradoxes arise in a cyclic history
extension of the TDP space with $C$ colors and $H$ histories if
and only if $C\mid H$.
\end{thm}

\begin{proof}
To show the absence of paradoxes for \textbf{all} initial conditions,
we must fully characterize how these initial conditions evolve in
order to derive consistency constraints for particles traveling along
CCCs. The absence of a paradox is equivalent to the particles traveling
along CCCs satisfying these constraints.

According to Corollary \ref{cor:Consistent-Positions}, positions
along all particle trajectories are consistent. Thus, we need only
show that the \textbf{colors }along these trajectories are consistent
as well. This analysis is greatly simplified by Lemma \ref{lem:Only-m-n},
which ensures that the only variables we need to consider when analyzing
the color evolution of particles in $J^{0}(\M')$ are $m$ and $n$,
the number of particles entering from the left ($-x$) and the right
($+x$) respectively. As noted in the proof of Lemma \ref{lem:Only-m-n},
the only particles traversing the time machine -- and thus the only
particles traveling along CCCs -- are confined exclusively to $J^{0}(\M')$,
so the relevant initial conditions for deriving consistency constraints
are entirely specified by $m$ and $n$.

We can more easily determine how the colors of these particles evolve
over the course of one history by identifying three zones in $J^{0}(\M')$
where there are group collisions. These zones are illustrated in Figure
\ref{fig:zones}. According to Corollary \ref{cor:All-histories-same},
the structure of the collisions is the same in each history, so characterizing
the evolution of particle colors in one history allows us to determine
this more broadly over $J^{0}(\M')$.

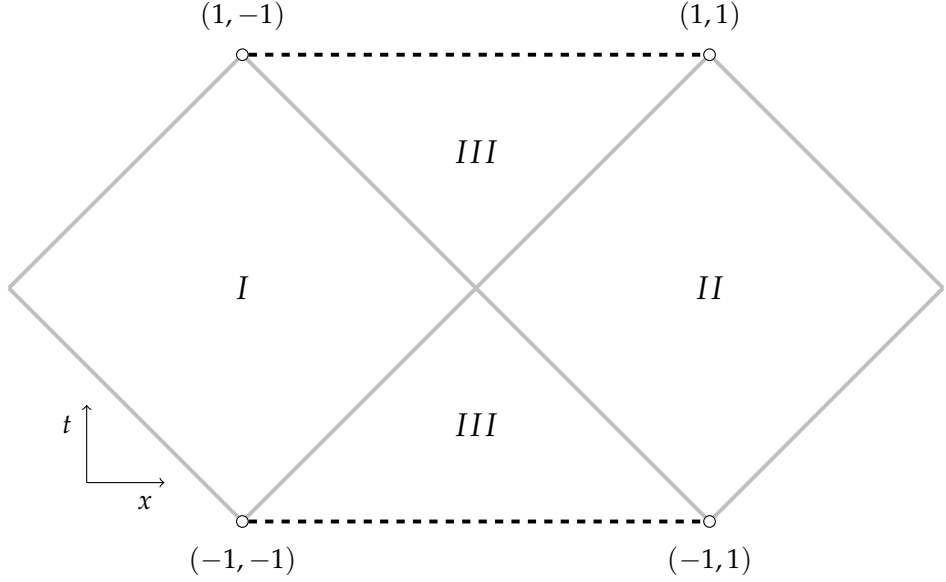
\begin{figure}[!tph]
\centering{}\begin{adjustbox}{width=0.8\textwidth}\input{Figure-zones.tex}\end{adjustbox}\caption{\label{fig:zones}A single history's causality-violating region can
be partitioned into three zones, each of which contains a group collision
of particles.}
\end{figure}

Since all the incoming particles collide such that they scatter away
from $J^{0}(\M')$ without ever traversing the wormhole, particles
approaching from $-x$ may participate only in collisions in zone
I and particles approaching from $+x$ may participate only in collisions
in zone II. Thus, $m$ particles leave zone III going in the $-x$
direction, participate in a group collision in zone I, and are scattered
back into zone III; similarly, $n$ particles leave zone III going
in the $+x$ direction, participate in a group collision in zone II,
and are scattered back into zone III. These $m+n$ particles are those
which follow CCCs, imposing consistency constraints that must be satisfied
to produce a legitimate solution to a given set of initial conditions.

The particles following CCCs collide in a group collision in zone
III. However, it is not initially clear that the structure of this
collision is the same as that of those in zones I and II: the group
collision is interrupted by a wormhole that flips the spatial ordering
of the particles. Nevertheless, we can treat this collision in the
same way as the others. This can be visualized by stacking a copy
of a single history's time machine region, flipped in $x$, on top
of itself, as shown in Figure \ref{fig:stacked}. This construction
illustrates that the zone III collision takes the same form as the
others, and we do not have to account for the wormhole's spatial flip
until after the collision, if we consider first the group collisions
in zones I and II, and then the zone III collision overlapping one
history and the next. 

\begin{figure}[!tph]
\begin{centering}
\begin{adjustbox}{width=0.66\textwidth}\input{Figure-stacked.tex}\end{adjustbox}
\par\end{centering}
\caption{\label{fig:stacked}Here, a \textbf{reflected }version of the $h=2$
causality-violating region is stacked on top of the $h=1$ causality-violating
region. These two regions lie in different spaces, as indicated by
the separate coordinate axes. However, this representation makes it
easy to see how particles evolve over multiple histories, and what
the consistency constraints are: that particles on the last wormhole
surface match those on the first one.}
\end{figure}
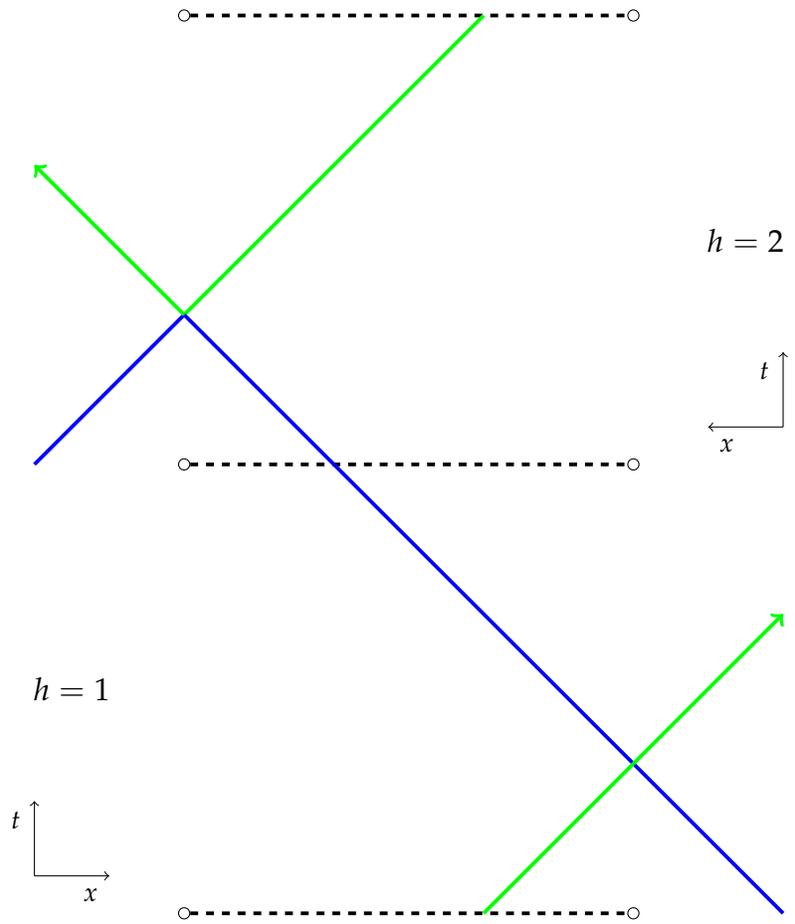

We are interested in whether there exists an assignment of colors
to the $m+n$ particles following CCCs that remains consistent after
the particles have traversed $J^{0}(\M')$ through all $H$ histories.
We begin by determining how these particle colors evolve over one
history, applying Lemma \ref{lem:Group-collision-evolution} to each
group collision.

Let $y_{k}$ be the color of particle number $k$, counting from the
left among those following CCCs, after the zone III collision has
taken place but before entering zone I or II. We first determine how
these values of $k$ relate to those used in Equation (\ref{eq:collision1})
and (\ref{eq:collision2}). For $k\leq m$, particle $y_{k}$ corresponds
to particle $x_{m+k}$ in an $(m,m)$ particle collision in zone I,
and then to particle $x_{k}$ in an $(m,n)$ particle collision in
zone III. For $k>m$, particle $y_{k}$ corresponds to particle $x_{k-m}$
in an $(n,n)$ particle collision in zone II, and then to particle
$x_{k}$ in an $(m,n)$ particle collision in zone III. Thus, evolving
through one history,

\begin{align}
y_{k}' & =y_{k}+\underbrace{\left(\begin{cases}
2(2m-(m+k))+1 & k\leq m,\\
2((k-m)-1)+1 & m<k\leq n,\\
2((k-m)-1)+1 & k>n,
\end{cases}\right)}_{\textrm{Zone I and II collisions}}+\underbrace{\left(\begin{cases}
2(k-1)+1 & k\leq m,\\
2m & m<k\leq n,\\
2(m+n-k)+1 & k>n,
\end{cases}\right)}_{\textrm{Zone III collision}}\nonumber \\
 & =y_{k}+\begin{cases}
2m & k\leq m,\\
2k-1 & m<k\leq n,\\
2n & k>n.
\end{cases}\label{eq:one-hist-evolution}
\end{align}

After passing through the wormhole, the spatial ordering of the particles
is reversed. Thus, particle number $k$ in one history becomes particle
number $m+n-k+1$ in the next. Particles previously satisfying $k\leq m$
now satisfy $k>n$, so these particles all increase in color by $2(m+n)$
after two histories. Also, since 
\[
\left(2k-1\right)+\left(2\left(m+n-k+1\right)-1\right)=2\left(m+n\right),
\]
particles satisfying $m<k\leq n$ also increase in color by $2(m+n)$.
Thus, if $H$ is even, the consistency constraint after traversing
these histories is:

\begin{equation}
y_{k}\equiv y_{k}+H(m+n)\mod C.\label{eq:evenConstraint}
\end{equation}

When $H=1$, the evolution of particle colors has already been given
by Equation \ref{eq:one-hist-evolution}. If $H$ is odd and $H>1$,
we can determine the evolution of particle colors by breaking the
evolution into that over $H-1$ histories (an even number) and that
over the last history. After traversing an odd number of histories,
$y_{k}'$ must equal $y_{m+n-k+1}$ for consistency. Thus, after $H$
traversals for odd $H$, the consistency constraint is:

\begin{equation}
y_{m+n-k+1}\equiv\left(y_{k}+\begin{cases}
(H+1)m+(H-1)n & k\leq m,\\
2k-1+(H-1)(m+n) & m<k\leq n\\
(H-1)m+(H+1)n & k>n,
\end{cases}\right)\mod C.\label{eq:oddConstraint}
\end{equation}

Note that both of these consistency requirements are in the group
$\Z_{C}$. When $k\neq\frac{m+n+1}{2}$ (that is, for all particles
except the middle particle when $m+n$ is odd) it implies that $2H(m+n)\equiv0\mod C$,
and is satisfied for all $m+n$ if and only if $C\mid2H$ (i.e. $C$
divides $2H$, or $2H$ is a multiple of $C$). However, when $m+n$
is odd, the $k=\frac{m+n+1}{2}$ condition implies that $H(m+n)\equiv0$,
and is satisfied for all $m+n$ if and only if $C\mid H$. The requirement
for even $H$ is satisfied if and only if $H(m+n)\equiv0\mod C$;
since all values of $m+n$ are possible, this requires that $C\mid H$. 

Thus, in general, $C\mid H$ is equivalent to the non-existence of
paradoxes for this system, since we can find a consistent solution
for the particle colors. When $C\nmid H$, setting $m=0$ and $n=1$
(corresponding to the natural extension of the paradox in Figure \ref{fig:Paradoxes})
provides a paradox, since the consistency constraints for both even
and odd $H$ cannot be satisfied in this case.
\end{proof}

\subsection{\label{subsec:Avoiding-Bootstrap-Paradoxes}Avoiding Bootstrap Paradoxes}

Although we have found conditions under which consistency paradoxes
can be avoided using a finite number of histories, these solutions
still have \textbf{bootstrap} paradoxes. This is because the particles
in this system are now of two separate types:
\begin{enumerate}
\item Particles that come from infinity also exit to infinity, never entering
the time machine;
\item Particles that emerge from the time machine in the past also enter
the time machine in the future, never leaving the causality-violating
region.
\end{enumerate}
The particles of the second type have no existence outside of the
causality-violating region (or outside CCCs), and therefore they are
``created from nothing'', which implies a bootstrap paradox. However,
we can avoid these bootstrap paradoxes by changing our interpretation
of the vertices.

In \cite{Krasnikov02} it is suggested for the case of $C=2$ that,
instead of as elastic collisions, the vertices can be interpreted
as intersections of \textbf{penetrable} particles, which simply flip
colors when they cross each other's paths. We can extend this interpretation
to our system for general $C$, where it involves a slightly more
complicated interaction: particles which pass through each other adopt
the other particle's color \textbf{plus $1$}. This re-interpretation
is sufficient to remove all bootstrap paradoxes, as depicted in Figure
\ref{fig:penetrable}.

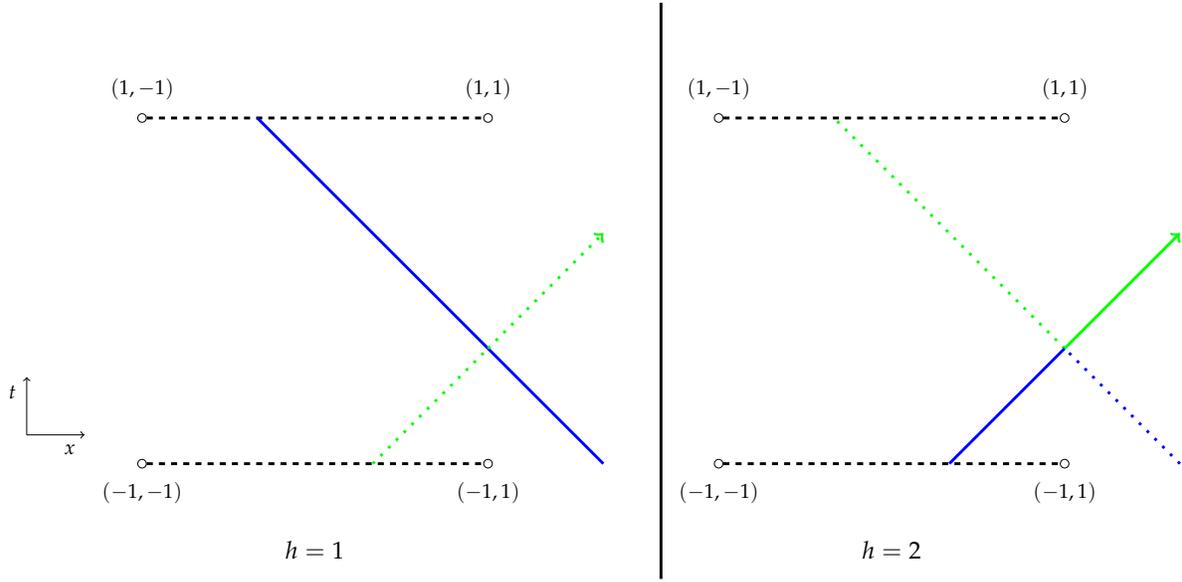
\begin{figure}[!tph]
\centering{}\begin{adjustbox}{width=\textwidth}\input{Figure-penetrable.tex}\end{adjustbox}\caption{\label{fig:penetrable}In this illustration, with $C=2$ and $H=2$,
one particle is solid and the other is dashed. The illustration demonstrates
an interpretation in which the particles do not collide; instead,
they \textbf{pass through }each other. This allows us to avoid a bootstrap
paradox. However, the same vertices in Figure \ref{fig:Four-vertices}
still apply.}
\end{figure}

The blue particle coming in from infinity in $h=1$ passes through
the green particle which came out of the time machine. To distinguish
the two particles, the first is indicated by a solid line while the
second is indicated by a dashed line. The particles interact using
vertex (c) in Figure \ref{fig:Four-vertices} -- which now means
that, instead of the particles changing both their directions and
colors, they change neither! The blue particle then goes through the
time machine and exits in $h=2$, where it meets its copy, which is
also blue (recall that the initial conditions are the same in each
history). The copy is indicated by a dashed line\footnote{The solid or dashed lines have no physical meaning, and they are not
properties of the particles themselves; they are just used in the
figure to distinguish one particle from the other in each vertex.
The actual physical property of the particle coming in from infinity
is that it is blue; the fact that it is solid in $h=1$ and dashed
in $h=2$ is simply for the purpose of distinguishing the particle
from its copy.}. The particle and its copy pass through each other, and they interact
using vertex (a) of Figure \ref{fig:Four-vertices}. The solid particle
changes its color to green, and goes out to infinity. The dashed particle
also changes its color to green, and goes into the time machine --
exiting from the past time machine in $h=1$, since the histories
are cyclic.

Neither of the particles actually follows a CCC, and both of them
have a clear start and end outside of the causality-violating region:
the solid particle enters from the right in $h=1$ and exits to the
right in $h=2$, while the dashed particle enters from the right in
$h=2$ and exits to the right in $h=1$. Thus, we avoid a bootstrap
paradox. This readily generalizes to larger values of $C$ and $H$.

\section{\label{sec:Analysis-of-Our}Analysis of Our Model}

\subsection{\label{subsec:Multiple-possibilities}Multiple Consistent Solutions}

Even in the base TDP space without multiple histories, not \textbf{every}
set of initial conditions necessarily causes a paradox. However, even
initial conditions which have consistent solutions still exhibit unusual
properties. Consider, for examples, the two solutions presented in
Figures \ref{fig:Consistent-1} and \ref{fig:Consistent-2}, where
the \textbf{same }initial conditions -- two blue particles coming
in, one from the left and one from the right -- lead to \textbf{two}
consistent color configurations. Thus, the evolution inside the causality-violating
region cannot necessarily be predicted from the initial conditions.
This situation does not usually appear in the absence of CCCs, as
classical physics is in general deterministic\footnote{However, see \cite{Norton}.}.
How will the universe ``decide'' which evolution to use? Choosing
a specific one would require additional assumptions to explain what
is special about that particular evolution.

\begin{figure}[!p]
\centering{}\begin{adjustbox}{width=0.66\textwidth}\input{Figure-TDP-non-paradox-1.tex}\end{adjustbox}\caption{\label{fig:Consistent-1}One of the two consistent solutions obtained
by sending particles toward the causality-violating region from both
sides.}
\end{figure}
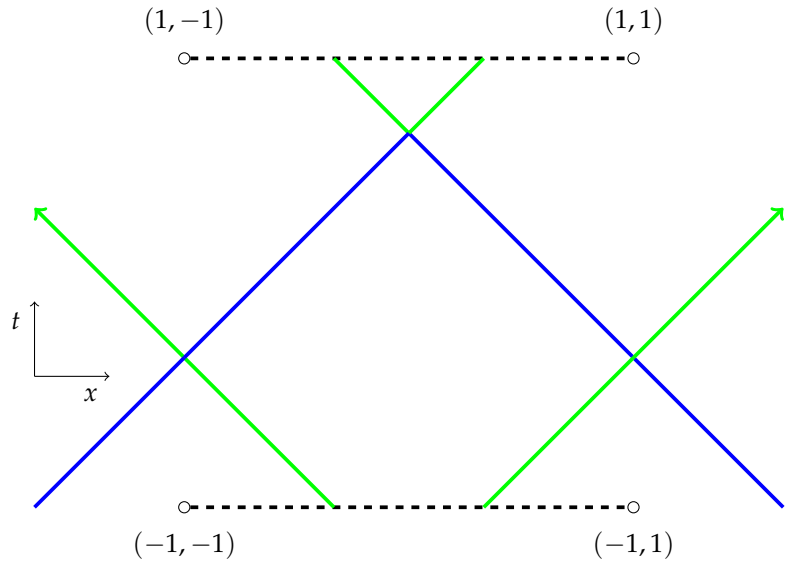

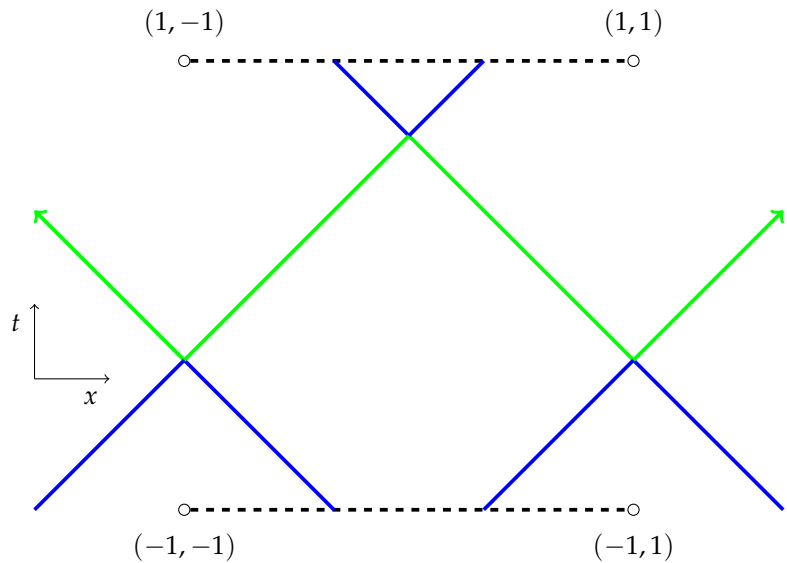
\begin{figure}[!p]
\centering{}\begin{adjustbox}{width=0.66\textwidth}\input{Figure-TDP-non-paradox-2.tex}\end{adjustbox}\caption{\label{fig:Consistent-2}The second of the two consistent solutions
obtained by sending particles toward the causality-violating region
from both sides. Note that the initial conditions and final outcomes
are the same as in Figure \ref{fig:Consistent-1} -- two blue particles
coming in and two green particles coming out -- but the evolution
inside the causality-violating region is different. Thus, evolution
in this region cannot be predicted.}
\end{figure}

The same situation also occurs generically in covering spaces of the
TDP space. Referring back to the systems of equations required to
ensure consistency in Section \ref{subsec:Analyzing-Particle-Collisions},
we can determine the number of free color variables, and consequently
the number of distinct consistent solutions. When $H$ is even, the
relevant constraint is Equation \ref{eq:evenConstraint}, where each
particle color is independent of the rest. Thus, when $H$ is even
there are $C^{m+n}$ possible solutions. When $H$ is odd, the relevant
constraint is Equation \ref{eq:oddConstraint}, where most equations
are coupled in pairs, giving $C^{\lceil\frac{m+n}{2}\rceil}$ solutions. 

Thus, the notion of multiple histories seems to arise from causality-violating
spacetimes in two distinct ways: first, when we extend the spacetime
to resolve paradoxes, and second, when multiple outcomes are compatible
with the same initial conditions. This second case seems more akin
to the ``worlds'' of the Everett interpretation than to the histories
solving time travel paradoxes, as these ``worlds'' represent distinct
possible outcomes for the same physical process rather than an outside
intervention due to the presence of a time machine. Confronted with
similar phenomena in a different spacetime, Echeverria, Klinkhammer,
and Thorne suggested in \cite{Consortium91} the possibility of resolving
the situation using a quantum mechanical sum-over-histories method.

\subsection{Revisiting Previous Histories and the Novikov Conjecture}

In Chapter \ref{sec:Finite}, we found that, assuming the number of
colors $C$ is finite, it is sufficient to have $C$ different \textbf{cyclic
}histories in order to resolve every possible paradox, both consistency
and bootstrap. In other words, contrary to popular opinion, one does
\textbf{not} need to prevent going back to previously visited histories
in order to avoid paradoxes.

Note that, since we allowed going back to the very first history,
the time machine will always emit particles from the future as soon
as it is created, which is what one would expect if the Novikov conjecture
is true, but not from a traditional multiple-histories scenario, where
the first history should, by definition, be the one where no one has
``yet'' traveled back in time. However, while the Novikov conjecture
was originally only applied to a single history, it can be applied
more generally, in principle, to larger spacetimes -- even to those
containing multiple histories.

The scenario where travel to the first history is possible therefore
\textbf{extends} the Novikov self-consistency conjecture to multiple
histories. Indeed, under the traditional Novikov conjecture, since
there is only one history, when we open the time machine at $t=-1$,
particles \textbf{must }come out since they went (will go) into the
time machine at $t=+1$. This is similar to how, in the traditional
Novikov conjecture scenario, a time traveler who goes back in time
to kill themselves has, in fact, \textbf{already }gone back and \textbf{already
}failed. There is no history where the time traveler did not go back
in time ``yet'', since there is only \textbf{one }history.

To illustrate this more precisely, consider a scenario where Alice
wants to travel back in time from 2020 to 1950 and kill her grandfather,
Bob, before he met her grandmother. Let us first assume that the Novikov
conjecture is correct, but there is only one history. Then in this
one history, in chronological order, Bob is born in 1930, Alice emerges
from a time machine in 1950 and tries to kill Bob -- but fails, Alice
is born in 1990, and Alice goes into a time machine in 2020. This
is a completely consistent chain of events, and there is no ``other''
universe or history where Alice did \textbf{not }travel back to 1950.

Next, let us assume that there are multiple histories, and that they
are cyclic all the way back to the first history. Then, again, there
is a completely self-consistent chain of events -- however, now it
encompasses more than one history. We will denote the year 2020 in
history A as 2020A, and so on, and we will similarly denote Alice
from history A as Alice A, and so on.

In history A, Bob A is born in 1930A, Alice B emerges from a time
machine in 1950A and tries to kill Bob A \textbf{by releasing a crocodile}
-- but fails, Alice A is born in 1990A, Bob A tells Alice A in 2010A
a story about a woman who looked remarkably like an older version
of her who tried to kill him back in 1950A by releasing a crocodile,
and Alice A goes into a time machine in 2020A determined to kill her
grandfather using another, more efficient method: dropping a piano
on him.

In history B, Bob B is born in 1930B, Alice A emerges from a time
machine in 1950B and tries to kill Bob B \textbf{by dropping a piano
on him} -- but fails, Alice B is born in 1990B, Bob B tells Alice
B in 2010B a story about a woman who looked remarkably like an older
version of her who tried to kill him back in 1950B by dropping a piano
on him, and Alice B goes into a time machine in 2020B determined to
kill her grandfather using another, more efficient method: releasing
a crocodile.

This is a Novikov-like scenario, but with two distinct histories which
are \textbf{not }self-consistent individually, since the murder attempts
in each history are different; when Alice B tries to kill Bob A by
releasing a crocodile, she is deliberately doing something that she
knows \textbf{not }to be consistent with her own history (B), as she
is trying to \textbf{change }history. Although she does manage to
change history from her perspective (into history A), Novikov's conjecture
still conspires to prevent her from changing it in an inconsistent
way; the combination of histories A and B together represents a completely
self-consistent chain of events, spanning two distinct histories.

These examples illustrate how paradox resolution using finite cyclic
histories leads to a novel hybrid scenario, with outcomes characteristic
of both one-history spacetimes satisfying the Novikov conjecture and
multiple-histories models with unlimited histories. With the Novikov
conjecture over only one history, there are closed causal curves,
and paradoxes are avoided when consistency can be enforced along these
curves. The price we have to pay is that the actions time travelers
take after they travel to the past must be predetermined, making time
travel essentially trivial. You can never go back in time to kill
Hitler, because there is only one history, and in this history Hitler
existed\footnote{Or maybe you go back in time to kill Hitler -- but fail, and this
near-death experience turns out to be what caused Hitler to become
an evil dictator in the first place!}.

On the other hand, with an infinite number of histories, paradoxes
are resolved by simply eliminating closed causal curves in the first
place. The price we have to pay is that, since CCCs do not exist,
this is not ``true'' time travel anymore\footnote{But see Footnote \ref{fn:Although-CCCs-are}.}.

Extending the Novikov conjecture to finite cyclic histories provides
a new middle ground for solving time travel paradoxes. The first case
occurs when $H=1$ and the second when $H$ is infinite. In between,
when $H$ is finite and greater than 1, we may have the existence
of closed causal curves over multiple histories while still satisfying
the Novikov conjecture, enabling ``true'' time travel along with
the ability to change history\footnote{\label{fn:Furthermore,-the-traditional}Furthermore, the traditional
Novikov scenario, with only one history, does not leave any room for
``free will'', since Alice cannot make any choice that will change
the past; if Alice already knows how her future self attempted to
kill Bob in the past, then she will simply \textbf{not be able to
choose} to try killing him in another way. However, with finite cyclic
histories, Alice \textbf{does}, in fact, have the capacity to change
history -- as Alice A and B did in the example above. If she ever
succeeds, then the chain of histories is simply terminated; however,
it is also possible that she fails every single time, in which case
the histories can (but do not necessarily have to) be cyclic. Thus,
this scenario provides at least the \textbf{illusion }of ``free will''.
Importantly, note that since different Alices exit the time machine
in each history, Alice B does not have a memory of what happened in
history A, so as far as she is concerned, she has the ``free will''
to do whatever she wants. More generally, there is never a situation
where Alice knows what is supposed to happen and finds out that she
simply does not have the ability to change it, which is the main issue
with Novikov's conjecture. Each history is a completely new history,
from Alice's perspective, with endless possibilities and nothing predetermined.}.

\subsection{\label{sec:Observable-Consequences}Observable Consequences}

We now have \textbf{four }different ways in which our universe might
resolve time travel paradoxes:
\begin{enumerate}
\item \textbf{The Hawking conjecture:} Time travel is simply impossible.
\item \textbf{The Novikov conjecture:} There is only one history, and it
can never be changed.
\item \textbf{Branching spacetime scenario:} Observers who travel back in
time find themselves in a new history and unable to go back to a previous
history. Furthermore, there is a unique first history.
\item \textbf{Covering space scenario:} There is no unique first history
and it is possible to return to a previous history when the number
of histories is finite -- as long as the Novikov conjecture applies
to the long closed causal curves which traverse all of the histories
(as opposed to each history individually).
\end{enumerate}
How may we experimentally determine which approach, if any, is realized
in our universe? First, if we successfully build a time machine, then
we have disproved the Hawking conjecture\footnote{As is usually the case, finding a counter-example to the conjecture
is much easier than proving that it is true in all cases. To \textbf{prove
}the Hawking conjecture, it is not enough to simply not succeed in
building a time machine, since it is always possible that a time machine
\textbf{can }be built, but we are just not skillful enough to build
it. The proof must therefore be a theoretical one; we must have access
to the most fundamental theory of physics (if such a theory exists...),
and use that theory to mathematically prove that a time machine \textbf{cannot
}be built, even in principle.}. Let us thus assume that it is indeed possible to build a time machine,
and discuss how to distinguish among options 2, 3, and 4. Consider
a simple experiment where Alice builds the time machine described
by the TDP space, which connects $t=+1$ to $t=-1$.

First, let us assume that Alice notices that another Alice did \textbf{not
}emerge from the time machine at $t=-1$. She then enters the time
machine at $t=+1$ and meets a copy of herself, who confirms that
the time is now $t=-1$. Both Alices now know that there are at least
two independent histories: the one where Alice did not exit the time
machine at $t=-1$, and the one where she did. Among the four models
we examine in this section, the Alices conclude that the \textbf{branching
spacetime} scenario must be the correct one, since Alice necessarily
came from the \textbf{first }history, $h=1$, and arrived at another
history, $h=2$.

Alternatively, let us assume that Alice (we shall now call her Alice
A) notices that another Alice (Alice B) \textbf{did }emerge from the
the time machine at $t=-1$. Then the Alices can try to change something
that Alice B remembers happening, which should be trivial assuming
that Alice B remembers everything that happened between $t=-1$ and
$t=+1$ in her history. For example, if Alice B remembers that she
said ``1'' then Alice A can try to say ``2'' instead.
\begin{itemize}
\item If they succeed in changing something, then they can conclude that
they live in a \textbf{covering space }scenario -- since there is
no first Alice, but also more than one history.
\item If they fail to do so, then they can suspect that they are in a \textbf{Novikov
conjecture }scenario.
\end{itemize}

\section{\label{sec:Discussion-and-Future}Summary and Future Plans}

In this paper, we introduced a (1+1)-dimensional model for a spacetime
with a time machine and multiple histories, and showed how time travel
paradoxes within this model are inevitable unless one allows for sufficiently
many histories. An infinite number of histories is certainly sufficient;
however, we also showed that a finite number of cyclic histories is
sufficient within our particular model, producing a variation of Novikov's
conjecture which spans multiple histories. This scenario contains
closed causal curves, unlike traditional multiple-histories resolutions,
while also allowing time travelers to actually change the past, unlike
Novikov's conjecture over only one history. Therefore, it provides
a good middle ground between the two. We also suggested how to experimentally
determine, at least in principle, whether our universe is described
by the Hawking conjecture, the Novikov conjecture, a branching spacetime
model, or a covering space model.

There are several important issues that we did not discuss here, including
the following:
\begin{itemize}
\item We did not provide an actual physical mechanism for creating new histories;
we merely assumed them, as is usually done in the literature.
\item We asserted that, in the TDP space, a time traveler moves from one
history to another while traversing the wormhole. However, we did
not develop a prescription for determining at which point along a
closed timelike or causal curve this transition between histories
happens in more general spacetimes. This question is of particular
concern in the case of more ``realistic'' time travel models, such
as those using warp drives or wormholes with non-zero throat length.
As time travel in this case involves traversing a non-zero distance,
it is unclear where exactly along this journey the new history should
be created. This problem becomes even more complicated when one considers
that closed curves, by definition, do not have a beginning or end!
\end{itemize}
We hope to address these issues in future work. Other intriguing avenues
of future research include generalizing our model in different ways,
such as the following:
\begin{itemize}
\item Formulating the model in 2+1 and 3+1 spacetime dimensions.
\item Employing realistic physical laws, ideally given by a well-defined
Lagrangian.
\item Allowing particles to travel along timelike paths in addition to null
paths.
\item Allowing additional time machines.
\item Allowing time machines to be turned on and off.
\end{itemize}
Multiple histories are, in our opinion, the most compelling of the
existing approaches for resolving time travel paradox. Hawking's conjecture
simply prevents time travel from happening in the first place, while
Novikov's conjecture allows time travel, but in an extremely limited
way, where the past cannot be changed and the time traveler cannot
exercise their ``free will''. If either conjecture is true, it would
make life much less interesting.

In contrast, the multiple-histories approach allows changing the past,
and at least the illusion of ``free will'' -- thus making the universe
considerably more exciting. In addition, it challenges many fundamental
notions in mathematics, physics, and philosophy, and opens up stimulating
new avenues of research. Yet, there is surprisingly little literature
about it. Furthermore, our presentation in this paper of a novel approach
-- the cyclic multiple-histories approach, which extends the Novikov
conjecture to multiple histories and exhibits hybrid behavior characteristic
of both the Novikov conjecture and multiple histories -- may provide
new and interesting ways in which time travel paradoxes can be discussed
and analyzed.

We hope that this paper will inspire mathematicians, physicists, and
philosophers to work on the formulation of a consistent and well-defined
framework for physics with multiple histories, both in relation to
time travel paradoxes and in other contexts, such as the Everett interpretation
of quantum mechanics.

\section{Acknowledgments}

B. S. would like to thank (in alphabetical order) Daniel C. Guariento,
Kasia Rejzner, and Rafael Sorkin for stimulating discussions. J. H.
would like to thank Matthew Fox for helpful discussions and Perimeter
Institute for providing the opportunity to perform this research as
part of its undergraduate summer program.

Research at Perimeter Institute is supported in part by the Government
of Canada through the Department of Innovation, Science and Economic
Development Canada and by the Province of Ontario through the Ministry
of Colleges and Universities.

\bibliographystyle{Utphys}
\phantomsection\addcontentsline{toc}{section}{\refname}\bibliography{TimeTravelParadoxes}

\end{document}

%% file: TikZstyles.tex

\tikzstyle{red dot}=[fill=red, draw=black, shape=circle]
\tikzstyle{singularity}=[fill=white, draw=black, shape=circle, inner sep=1.5pt]
\tikzstyle{black dot}=[fill=black, draw=black, shape=circle, inner sep=2pt]
\tikzstyle{small black dot}=[fill=black, draw=black, shape=circle, inner sep=1pt]
\tikzstyle{smaller black dot}=[fill=black, draw=black, shape=circle, inner sep=0.5pt]
\tikzstyle{ball}=[fill={rgb,255: red,191; green,191; blue,191}, draw={rgb,255: red,128; green,128; blue,128}, shape=circle, inner sep=10pt, dashed]

\tikzstyle{wormhole}=[-, dashed, line width=0.5mm]
\tikzstyle{particle 1}=[-, line width=0.5mm]
\tikzstyle{particle tip 1}=[->, line width=0.5mm]
\tikzstyle{particle 2}=[-, draw={rgb,255: red,191; green,191; blue,191}, line width=0.5mm]
\tikzstyle{particle tip 2}=[draw={rgb,255: red,191; green,191; blue,191}, ->, line width=0.5mm]
\tikzstyle{invalid particle}=[-, draw=red, line width=0.5mm]
\tikzstyle{axis}=[->]
\tikzstyle{particle double 1}=[-, line width=0.5mm, loosely dotted]
\tikzstyle{particle double 2}=[-, draw={rgb,255: red,191; green,191; blue,191}, line width=0.5mm, loosely dotted]
\tikzstyle{particle double tip 1}=[->, line width=0.5mm, loosely dotted]
\tikzstyle{particle double tip 2}=[->, draw={rgb,255: red,191; green,191; blue,191}, line width=0.5mm, loosely dotted]

\tikzstyle{particle color 1}=[-, line width=0.5mm, blue]
\tikzstyle{particle tip color 1}=[->, line width=0.5mm, blue]
\tikzstyle{particle color 2}=[-, line width=0.5mm, green]
\tikzstyle{particle tip color 2}=[->, line width=0.5mm, green]
\tikzstyle{particle color 3}=[-, line width=0.5mm, magenta]
\tikzstyle{particle tip color 3}=[->, line width=0.5mm, magenta]
\tikzstyle{particle color 4}=[-, line width=0.5mm, orange]
\tikzstyle{particle tip color 4}=[->, line width=0.5mm, orange]
\tikzstyle{particle color 5}=[-, line width=0.5mm, yellow]
\tikzstyle{particle tip color 5}=[->, line width=0.5mm, yellow]
\tikzstyle{particle double color 1}=[-, line width=0.5mm, loosely dotted, blue]
\tikzstyle{particle double color 2}=[-, line width=0.5mm, loosely dotted, green]
\tikzstyle{particle double tip color 1}=[->, line width=0.5mm, loosely dotted, blue]
\tikzstyle{particle double tip color 2}=[->, line width=0.5mm, loosely dotted, green]

%% file: Figure-DP-space.tex
\begin{tikzpicture}
	\begin{pgfonlayer}{nodelayer}
		\node [style=singularity] (0) at (-3, 3) {};
		\node [style=singularity] (1) at (3, 3) {};
		\node [style=singularity] (2) at (-3, -3) {};
		\node [style=singularity] (3) at (3, -3) {};
		\node [style=none] (4) at (-1, 3) {};
		\node [style=none] (5) at (5, -3) {};
		\node [style=none] (6) at (-1, -3) {};
		\node [style=none] (7) at (-2.5, -1.5) {};
		\node [style=none] (8) at (-3, 3.5) {$(1,-1)$};
		\node [style=none] (9) at (3, 3.5) {$(1,1)$};
		\node [style=none] (10) at (-3, -3.5) {$(-1,-1)$};
		\node [style=none] (11) at (3, -3.5) {$(-1,1)$};
		\node [style=none] (12) at (-5, -2.5) {};
		\node [style=none] (13) at (-5, -1.5) {};
		\node [style=none] (14) at (-4, -2.5) {};
		\node [style=none] (18) at (-4.25, -2.75) {$x$};
		\node [style=none] (19) at (-5.25, -1.75) {$t$};
	\end{pgfonlayer}
	\begin{pgfonlayer}{edgelayer}
		\draw [style=wormhole] (0) to (1);
		\draw [style=wormhole] (2) to (3);
		\draw [style=particle color 1] (5.center) to (4.center);
		\draw [style=particle tip color 1] (6.center) to (7.center);
		\draw [style=axis] (12.center) to (13.center);
		\draw [style=axis] (12.center) to (14.center);
	\end{pgfonlayer}
\end{tikzpicture}

%% file: Figure-TDP-space.tex
\begin{tikzpicture}
	\begin{pgfonlayer}{nodelayer}
		\node [style=singularity] (0) at (-3, 3) {};
		\node [style=singularity] (1) at (3, 3) {};
		\node [style=singularity] (2) at (-3, -3) {};
		\node [style=singularity] (3) at (3, -3) {};
		\node [style=none] (4) at (-1, 3) {};
		\node [style=none] (5) at (5, -3) {};
		\node [style=none] (6) at (1, -3) {};
		\node [style=none] (7) at (2.5, -1.5) {};
		\node [style=none] (8) at (-3, 3.5) {$(1,-1)$};
		\node [style=none] (9) at (3, 3.5) {$(1,1)$};
		\node [style=none] (10) at (-3, -3.5) {$(-1,-1)$};
		\node [style=none] (11) at (3, -3.5) {$(-1,1)$};
		\node [style=none] (12) at (-5, -2.5) {};
		\node [style=none] (13) at (-5, -1.5) {};
		\node [style=none] (14) at (-4, -2.5) {};
		\node [style=none] (15) at (-4.25, -2.75) {$x$};
		\node [style=none] (16) at (-5.25, -1.75) {$t$};
	\end{pgfonlayer}
	\begin{pgfonlayer}{edgelayer}
		\draw [style=wormhole] (0) to (1);
		\draw [style=wormhole] (2) to (3);
		\draw [style=particle color 1] (5.center) to (4.center);
		\draw [style=particle tip color 1] (6.center) to (7.center);
		\draw [style=axis] (12.center) to (13.center);
		\draw [style=axis] (12.center) to (14.center);
	\end{pgfonlayer}
\end{tikzpicture}

%% file: Figure-causality-violating.tex
\begin{tikzpicture}
	\begin{pgfonlayer}{nodelayer}
		\node [style=singularity] (0) at (-3, 3) {};
		\node [style=singularity] (1) at (3, 3) {};
		\node [style=singularity] (2) at (-3, -3) {};
		\node [style=singularity] (3) at (3, -3) {};
		\node [style=none] (8) at (-3, 3.5) {$(1,-1)$};
		\node [style=none] (9) at (3, 3.5) {$(1,1)$};
		\node [style=none] (10) at (-3, -3.5) {$(-1,-1)$};
		\node [style=none] (11) at (3, -3.5) {$(-1,1)$};
		\node [style=none] (12) at (-5, -2.5) {};
		\node [style=none] (13) at (-5, -1.5) {};
		\node [style=none] (14) at (-4, -2.5) {};
		\node [style=none] (15) at (-4.25, -2.75) {$x$};
		\node [style=none] (16) at (-5.25, -1.75) {$t$};
		\node [style=none] (17) at (6, 0) {};
		\node [style=none] (18) at (-6, 0) {};
		\node [style=none] (21) at (0, 0) {\large $J^0(\M)$};
		\node [style=none] (22) at (-5.5, 2.5) {\large $\M$};
		\node [style=none] (23) at (-5.25, -3.5) {};
		\node [style=none] (24) at (5, -3) {};
		\node [style=none] (25) at (8, -5) {};
		\node [style=none] (26) at (5.5, -2.5) {};
		\node [style=none] (27) at (-6.75, -5) {};
		\node [style=none] (28) at (7, -5) {};
		\node [style=none] (29) at (-9, -5) {};
		\node [style=none] (30) at (9, -5) {};
	\end{pgfonlayer}
	\begin{pgfonlayer}{edgelayer}
		\draw [style=wormhole] (0) to (1);
		\draw [style=wormhole] (2) to (3);
		\draw [style=axis] (12.center) to (13.center);
		\draw [style=axis] (12.center) to (14.center);
		\draw [style=particle 1] (3) to (17.center);
		\draw [style=particle 1] (17.center) to (1);
		\draw [style=particle 1] (0) to (18.center);
		\draw [style=particle 1] (18.center) to (2);
		\draw [style=particle tip color 1] (27.center) to (23.center);
		\draw [style=particle tip color 2] (28.center) to (24.center);
		\draw [style=particle tip color 1] (25.center) to (26.center);
		\draw [style=particle 2] (25.center) to (28.center);
		\draw [style=particle 2] (28.center) to (27.center);
		\draw [style=particle 2] (30.center) to (25.center);
		\draw [style=particle 2] (27.center) to (29.center);
	\end{pgfonlayer}
\end{tikzpicture}

%% file: Figure-vertex1.tex
\begin{tikzpicture}
	\begin{pgfonlayer}{nodelayer}
		\node [style=none] (0) at (2, -2) {};
		\node [style=none] (1) at (-2, 2) {};
		\node [style=none] (2) at (2, 2) {};
		\node [style=none] (3) at (-2, -2) {};
		\node [style=none] (4) at (0, 0) {};
	\end{pgfonlayer}
	\begin{pgfonlayer}{edgelayer}
		\draw [style=particle color 1] (3.center) to (4.center);
		\draw [style=particle color 1] (0.center) to (4.center);
		\draw [style=particle tip color 2] (4.center) to (2.center);
		\draw [style=particle tip color 2] (4.center) to (1.center);
	\end{pgfonlayer}
\end{tikzpicture}

%% file: Figure-vertex2.tex
\begin{tikzpicture}
	\begin{pgfonlayer}{nodelayer}
		\node [style=none] (0) at (2, -2) {};
		\node [style=none] (1) at (-2, 2) {};
		\node [style=none] (2) at (2, 2) {};
		\node [style=none] (3) at (-2, -2) {};
		\node [style=none] (4) at (0, 0) {};
	\end{pgfonlayer}
	\begin{pgfonlayer}{edgelayer}
		\draw [style=particle color 1] (3.center) to (4.center);
		\draw [style=particle color 2] (0.center) to (4.center);
		\draw [style=particle tip color 1] (4.center) to (2.center);
		\draw [style=particle tip color 2] (4.center) to (1.center);
	\end{pgfonlayer}
\end{tikzpicture}

%% file: Figure-vertex3.tex
\begin{tikzpicture}
	\begin{pgfonlayer}{nodelayer}
		\node [style=none] (0) at (2, -2) {};
		\node [style=none] (1) at (-2, 2) {};
		\node [style=none] (2) at (2, 2) {};
		\node [style=none] (3) at (-2, -2) {};
		\node [style=none] (4) at (0, 0) {};
	\end{pgfonlayer}
	\begin{pgfonlayer}{edgelayer}
		\draw [style=particle color 2] (3.center) to (4.center);
		\draw [style=particle color 1] (0.center) to (4.center);
		\draw [style=particle tip color 2] (4.center) to (2.center);
		\draw [style=particle tip color 1] (4.center) to (1.center);
	\end{pgfonlayer}
\end{tikzpicture}

%% file: Figure-vertex4.tex
\begin{tikzpicture}
	\begin{pgfonlayer}{nodelayer}
		\node [style=none] (0) at (2, -2) {};
		\node [style=none] (1) at (-2, 2) {};
		\node [style=none] (2) at (2, 2) {};
		\node [style=none] (3) at (-2, -2) {};
		\node [style=none] (4) at (0, 0) {};
	\end{pgfonlayer}
	\begin{pgfonlayer}{edgelayer}
		\draw [style=particle color 2] (3.center) to (4.center);
		\draw [style=particle color 2] (0.center) to (4.center);
		\draw [style=particle tip color 1] (4.center) to (2.center);
		\draw [style=particle tip color 1] (4.center) to (1.center);
	\end{pgfonlayer}
\end{tikzpicture}

%% file: Figure-TDP-paradox.tex
\begin{tikzpicture}
	\begin{pgfonlayer}{nodelayer}
		\node [style=singularity] (0) at (-3, 3) {};
		\node [style=singularity] (1) at (3, 3) {};
		\node [style=singularity] (2) at (-3, -3) {};
		\node [style=singularity] (3) at (3, -3) {};
		\node [style=none] (5) at (5, -3) {};
		\node [style=none] (6) at (1, -3) {};
		\node [style=none] (7) at (3, -1) {};
		\node [style=none] (8) at (-1, 3) {};
		\node [style=none] (9) at (5, 1) {};
		\node [style=none] (10) at (-5, -2.5) {};
		\node [style=none] (11) at (-5, -1.5) {};
		\node [style=none] (12) at (-4, -2.5) {};
		\node [style=none] (13) at (-4.25, -2.75) {$x$};
		\node [style=none] (14) at (-5.25, -1.75) {$t$};
		\node [style=none] (15) at (-3, 3.5) {$(1,-1)$};
		\node [style=none] (16) at (3, 3.5) {$(1,1)$};
		\node [style=none] (17) at (-3, -3.5) {$(-1,-1)$};
		\node [style=none] (18) at (3, -3.5) {$(-1,1)$};
	\end{pgfonlayer}
	\begin{pgfonlayer}{edgelayer}
		\draw [style=wormhole] (0) to (1);
		\draw [style=wormhole] (2) to (3);
		\draw [style=particle 2] (6.center) to (7.center);
		\draw [style=particle color 1] (5.center) to (7.center);
		\draw [style=particle 2] (7.center) to (8.center);
		\draw [style=particle tip color 2, in=225, out=45] (7.center) to (9.center);
		\draw [style=axis] (10.center) to (11.center);
		\draw [style=axis] (10.center) to (12.center);
	\end{pgfonlayer}
\end{tikzpicture}

%% file: Figure-vertex-general.tex
\begin{tikzpicture}
	\begin{pgfonlayer}{nodelayer}
		\node [style=none] (0) at (2, -2) {};
		\node [style=none] (1) at (-2, 2) {};
		\node [style=none] (2) at (2, 2) {};
		\node [style=none] (3) at (-2, -2) {};
		\node [style=none] (4) at (0, 0) {};
		\node [style=none] (5) at (-2.5, -2.5) {$c$};
		\node [style=none] (6) at (2.5, -2.5) {$c'$};
		\node [style=none] (7) at (-2.5, 2.5) {$c+1$};
		\node [style=none] (8) at (2.5, 2.5) {$c'+1$};
	\end{pgfonlayer}
	\begin{pgfonlayer}{edgelayer}
		\draw [style=particle color 1, in=225, out=45] (3.center) to (4.center);
		\draw [style=particle color 2] (0.center) to (4.center);
		\draw [style=particle tip color 3] (4.center) to (2.center);
		\draw [style=particle tip color 4] (4.center) to (1.center);
	\end{pgfonlayer}
\end{tikzpicture}

%% file: Figure-vertex-cpt.tex
\begin{tikzpicture}
	\begin{pgfonlayer}{nodelayer}
		\node [style=none] (0) at (2, -2) {};
		\node [style=none] (1) at (-2, 2) {};
		\node [style=none] (2) at (2, 2) {};
		\node [style=none] (3) at (-2, -2) {};
		\node [style=none] (4) at (0, 0) {};
		\node [style=none] (5) at (-2.5, -2.5) {$-c'-1$};
		\node [style=none] (6) at (2.5, -2.5) {$-c-1$};
		\node [style=none] (7) at (-2.5, 2.5) {$-c'$};
		\node [style=none] (8) at (2.5, 2.5) {$-c$};
	\end{pgfonlayer}
	\begin{pgfonlayer}{edgelayer}
		\draw [style=particle color 4, in=225, out=45] (3.center) to (4.center);
		\draw [style=particle color 3] (0.center) to (4.center);
		\draw [style=particle tip color 1] (4.center) to (2.center);
		\draw [style=particle tip color 2] (4.center) to (1.center);
	\end{pgfonlayer}
\end{tikzpicture}

%% file: Figure-infinite-branching-solution.tex
\begin{tikzpicture}
	\begin{pgfonlayer}{nodelayer}
		\node [style=singularity] (0) at (-9, 3) {};
		\node [style=singularity] (1) at (-3, 3) {};
		\node [style=none] (4) at (-1, -3) {};
		\node [style=none] (6) at (-7, 3) {};
		\node [style=none] (9) at (-11, -2.5) {};
		\node [style=none] (10) at (-11, -1.5) {};
		\node [style=none] (11) at (-10, -2.5) {};
		\node [style=none] (12) at (-10.25, -2.75) {$x$};
		\node [style=none] (13) at (-11.25, -1.75) {$t$};
		\node [style=none] (14) at (-9, 3.5) {$(1,-1)$};
		\node [style=none] (15) at (-3, 3.5) {$(1,1)$};
		\node [style=singularity] (18) at (1, 3) {};
		\node [style=singularity] (19) at (7, 3) {};
		\node [style=singularity] (20) at (1, -3) {};
		\node [style=singularity] (21) at (7, -3) {};
		\node [style=none] (23) at (9, -3) {};
		\node [style=none] (24) at (5, -3) {};
		\node [style=none] (28) at (7, -1) {};
		\node [style=none] (30) at (9, 1) {};
		\node [style=none] (32) at (3, 3) {};
		\node [style=none] (38) at (1, 3.5) {$(1,-1)$};
		\node [style=none] (39) at (7, 3.5) {$(1,1)$};
		\node [style=none] (40) at (1, -3.5) {$(-1,-1)$};
		\node [style=none] (41) at (7, -3.5) {$(-1,1)$};
		\node [style=none] (42) at (0, 5) {};
		\node [style=none] (43) at (0, -5) {};
		\node [style=none] (44) at (10, -5) {};
		\node [style=none] (45) at (10, 5) {};
		\node [style=none] (47) at (4, -4.5) {\large $h=2$};
		\node [style=none] (48) at (-6, -4.5) {\large $h=1$};
		\node [style=black dot] (49) at (11, 0) {};
		\node [style=black dot] (50) at (12, 0) {};
		\node [style=black dot] (51) at (13, 0) {};
		\node [style=none] (52) at (12, -4.5) {\large $h>2$};
	\end{pgfonlayer}
	\begin{pgfonlayer}{edgelayer}
		\draw [style=wormhole] (0) to (1);
		\draw [style=particle color 1] (4.center) to (6.center);
		\draw [style=axis] (9.center) to (10.center);
		\draw [style=axis] (9.center) to (11.center);
		\draw [style=wormhole] (18) to (19);
		\draw [style=wormhole] (20) to (21);
		\draw [style=particle color 1] (23.center) to (28.center);
		\draw [style=particle color 2] (28.center) to (32.center);
		\draw [style=particle color 1] (24.center) to (28.center);
		\draw [style=particle tip color 2] (28.center) to (30.center);
		\draw [style=particle 1] (42.center) to (43.center);
		\draw [style=particle 1] (45.center) to (44.center);
	\end{pgfonlayer}
\end{tikzpicture}

%% file: Figure-infinite-covering-solution.tex
\begin{tikzpicture}
	\begin{pgfonlayer}{nodelayer}
		\node [style=singularity] (0) at (-9, 3) {};
		\node [style=singularity] (1) at (-3, 3) {};
		\node [style=singularity] (2) at (-9, -3) {};
		\node [style=singularity] (3) at (-3, -3) {};
		\node [style=none] (6) at (-11, -2.5) {};
		\node [style=none] (7) at (-11, -1.5) {};
		\node [style=none] (8) at (-10, -2.5) {};
		\node [style=none] (9) at (-10.25, -2.75) {$x$};
		\node [style=none] (10) at (-11.25, -1.75) {$t$};
		\node [style=none] (11) at (-9, 3.5) {$(1,-1)$};
		\node [style=none] (12) at (-3, 3.5) {$(1,1)$};
		\node [style=none] (13) at (-9, -3.5) {$(-1,-1)$};
		\node [style=none] (14) at (-3, -3.5) {$(-1,1)$};
		\node [style=singularity] (15) at (1, 3) {};
		\node [style=singularity] (16) at (7, 3) {};
		\node [style=singularity] (17) at (1, -3) {};
		\node [style=singularity] (18) at (7, -3) {};
		\node [style=none] (19) at (9, -3) {};
		\node [style=none] (20) at (5, -3) {};
		\node [style=none] (21) at (7, -1) {};
		\node [style=none] (22) at (9, 1) {};
		\node [style=none] (23) at (3, 3) {};
		\node [style=none] (24) at (1, 3.5) {$(1,-1)$};
		\node [style=none] (25) at (7, 3.5) {$(1,1)$};
		\node [style=none] (26) at (1, -3.5) {$(-1,-1)$};
		\node [style=none] (27) at (7, -3.5) {$(-1,1)$};
		\node [style=none] (28) at (0, 5) {};
		\node [style=none] (29) at (0, -5) {};
		\node [style=none] (32) at (4, -4.5) {\large $h=k+1$};
		\node [style=none] (33) at (-6, -4.5) {\large $h=k$};
		\node [style=none] (34) at (-1, -3) {};
		\node [style=none] (35) at (-5, -3) {};
		\node [style=none] (36) at (-3, -1) {};
		\node [style=none] (37) at (-1, 1) {};
		\node [style=none] (38) at (-7, 3) {};
		\node [style=none] (44) at (10, -5) {};
		\node [style=none] (45) at (10, 5) {};
		\node [style=none] (46) at (12, -4.5) {\large $h>k+1$};
		\node [style=none] (47) at (-14, -4.5) {\large $h<k$};
		\node [style=black dot] (49) at (11, 0) {};
		\node [style=black dot] (50) at (12, 0) {};
		\node [style=black dot] (51) at (13, 0) {};
		\node [style=none] (52) at (-12, -5) {};
		\node [style=none] (53) at (-12, 5) {};
		\node [style=black dot] (54) at (-13, 0) {};
		\node [style=black dot] (55) at (-14, 0) {};
		\node [style=black dot] (56) at (-15, 0) {};
	\end{pgfonlayer}
	\begin{pgfonlayer}{edgelayer}
		\draw [style=wormhole] (0) to (1);
		\draw [style=wormhole] (2) to (3);
		\draw [style=axis] (6.center) to (7.center);
		\draw [style=axis] (6.center) to (8.center);
		\draw [style=wormhole] (15) to (16);
		\draw [style=wormhole] (17) to (18);
		\draw [style=particle color 1] (19.center) to (21.center);
		\draw [style=particle color 1] (21.center) to (23.center);
		\draw [style=particle color 3] (20.center) to (21.center);
		\draw [style=particle tip color 4] (21.center) to (22.center);
		\draw [style=particle 1] (28.center) to (29.center);
		\draw [style=particle color 1] (34.center) to (36.center);
		\draw [style=particle color 3] (36.center) to (38.center);
		\draw [style=particle color 2] (35.center) to (36.center);
		\draw [style=particle tip color 4] (36.center) to (37.center);
		\draw [style=particle 1] (45.center) to (44.center);
		\draw [style=particle 1] (52.center) to (53.center);
	\end{pgfonlayer}
\end{tikzpicture}

%% file: Figure-open-set-1.tex
\begin{tikzpicture}
	\begin{pgfonlayer}{nodelayer}
		\node [style=ball] (21) at (-1, 3) {};
		\node [style=ball] (22) at (1, -3) {};
		\node [style=singularity] (0) at (-3, 3) {};
		\node [style=singularity] (1) at (3, 3) {};
		\node [style=singularity] (2) at (-3, -3) {};
		\node [style=singularity] (3) at (3, -3) {};
		\node [style=none] (8) at (-3, 3.5) {\small $(1,-1)$};
		\node [style=none] (9) at (3, 3.5) {\small $(1,1)$};
		\node [style=none] (10) at (-3, -3.5) {\small $(-1,-1)$};
		\node [style=none] (11) at (3, -3.5) {\small $(-1,1)$};
		\node [style=none] (12) at (-5, -2.5) {};
		\node [style=none] (13) at (-5, -1.5) {};
		\node [style=none] (14) at (-4, -2.5) {};
		\node [style=none] (15) at (-4.25, -2.75) {\small $x$};
		\node [style=none] (16) at (-5.25, -1.75) {\small $t$};
		\node [style=none] (17) at (-0.825, 2.85) {\footnotesize $m$};
		\node [style=none] (18) at (1.175, -3.15) {\footnotesize $m$};
		\node [style=smaller black dot] (19) at (-1, 3) {};
		\node [style=smaller black dot] (20) at (1, -3) {};
		\node [style=none] (23) at (-1.55, 3.55) { $U^+$};
		\node [style=none] (24) at (0.45, -2.45) { $U^-$};
		\node [style=none] (25) at (5, -3) {};
	\end{pgfonlayer}
	\begin{pgfonlayer}{edgelayer}
		\draw [style=wormhole] (0) to (1);
		\draw [style=wormhole] (2) to (3);
		\draw [style=axis] (12.center) to (13.center);
		\draw [style=axis] (12.center) to (14.center);
	\end{pgfonlayer}
\end{tikzpicture}

%% file: Figure-open-set-2.tex
\begin{tikzpicture}
	\begin{pgfonlayer}{nodelayer}
		\node [style=singularity] (0) at (-9, 3) {};
		\node [style=singularity] (1) at (-3, 3) {};
		\node [style=singularity] (2) at (-9, -3) {};
		\node [style=singularity] (3) at (-3, -3) {};
		\node [style=none] (6) at (-11, -2.5) {};
		\node [style=none] (7) at (-11, -1.5) {};
		\node [style=none] (8) at (-10, -2.5) {};
		\node [style=none] (9) at (-10.25, -2.75) {$x$};
		\node [style=none] (10) at (-11.25, -1.75) {$t$};
		\node [style=none] (11) at (-9, 3.5) {$(1,-1)$};
		\node [style=none] (12) at (-3, 3.5) {$(1,1)$};
		\node [style=none] (13) at (-9, -3.5) {$(-1,-1)$};
		\node [style=none] (14) at (-3, -3.5) {$(-1,1)$};
		\node [style=singularity] (15) at (1, 3) {};
		\node [style=singularity] (16) at (7, 3) {};
		\node [style=singularity] (17) at (1, -3) {};
		\node [style=singularity] (18) at (7, -3) {};
		\node [style=ball] (23) at (3, 3) {};
		\node [style=none] (24) at (1, 3.5) {$(1,-1)$};
		\node [style=none] (25) at (7, 3.5) {$(1,1)$};
		\node [style=none] (26) at (1, -3.5) {$(-1,-1)$};
		\node [style=none] (27) at (7, -3.5) {$(-1,1)$};
		\node [style=none] (28) at (0, 5) {};
		\node [style=none] (29) at (0, -5) {};
		\node [style=none] (32) at (4, -4.5) {\large $h=k+1$};
		\node [style=none] (33) at (-6, -4.5) {\large $h=k$};
		\node [style=ball] (38) at (-7, 3) {};
		\node [style=none] (44) at (10, -5) {};
		\node [style=none] (45) at (10, 5) {};
		\node [style=none] (46) at (12, -4.5) {\large $h>k+1$};
		\node [style=none] (47) at (-14, -4.5) {\large $h<k$};
		\node [style=black dot] (49) at (11, 0) {};
		\node [style=black dot] (50) at (12, 0) {};
		\node [style=black dot] (51) at (13, 0) {};
		\node [style=none] (52) at (-12, -5) {};
		\node [style=none] (53) at (-12, 5) {};
		\node [style=black dot] (54) at (-13, 0) {};
		\node [style=black dot] (55) at (-14, 0) {};
		\node [style=black dot] (56) at (-15, 0) {};
		\node [style=ball] (58) at (-5, -3) {};
		\node [style=ball] (59) at (5, -3) {};
		\node [style=small black dot] (60) at (3, 3) {};
		\node [style=small black dot] (61) at (5, -3) {};
		\node [style=small black dot] (62) at (-5, -3) {};
		\node [style=small black dot] (63) at (-7, 3) {};
		\node [style=none] (64) at (-7, 2.775) {\footnotesize $m_{k+1}$};
		\node [style=none] (65) at (-5, -3.225) {\footnotesize $m_{k}$};
		\node [style=none] (66) at (3, 2.775) {\footnotesize $m_{k+2}$};
		\node [style=none] (67) at (5, -3.225) {\footnotesize $m_{k+1}$};
		\node [style=none] (68) at (-7.6, 3.7) {$U^+_k$};
		\node [style=none] (69) at (-5.6, -2.3) {$U^-_k$};
		\node [style=none] (70) at (2.3, 3.7) {$U^+_{k+1}$};
		\node [style=none] (71) at (4.3, -2.3) {$U^-_{k+1}$};
	\end{pgfonlayer}
	\begin{pgfonlayer}{edgelayer}
		\draw [style=wormhole] (0) to (1);
		\draw [style=wormhole] (2) to (3);
		\draw [style=axis] (6.center) to (7.center);
		\draw [style=axis] (6.center) to (8.center);
		\draw [style=wormhole] (15) to (16);
		\draw [style=wormhole] (17) to (18);
		\draw [style=particle 1] (28.center) to (29.center);
		\draw [style=particle 1] (45.center) to (44.center);
		\draw [style=particle 1] (52.center) to (53.center);
	\end{pgfonlayer}
\end{tikzpicture}

%% file: Figure-finite-solution.tex
\begin{tikzpicture}
	\begin{pgfonlayer}{nodelayer}
		\node [style=singularity] (0) at (-9, 3) {};
		\node [style=singularity] (1) at (-3, 3) {};
		\node [style=singularity] (2) at (-9, -3) {};
		\node [style=singularity] (3) at (-3, -3) {};
		\node [style=none] (6) at (-11, -2.5) {};
		\node [style=none] (7) at (-11, -1.5) {};
		\node [style=none] (8) at (-10, -2.5) {};
		\node [style=none] (9) at (-10.25, -2.75) {$x$};
		\node [style=none] (10) at (-11.25, -1.75) {$t$};
		\node [style=none] (11) at (-9, 3.5) {$(1,-1)$};
		\node [style=none] (12) at (-3, 3.5) {$(1,1)$};
		\node [style=none] (13) at (-9, -3.5) {$(-1,-1)$};
		\node [style=none] (14) at (-3, -3.5) {$(-1,1)$};
		\node [style=singularity] (15) at (1, 3) {};
		\node [style=singularity] (16) at (7, 3) {};
		\node [style=singularity] (17) at (1, -3) {};
		\node [style=singularity] (18) at (7, -3) {};
		\node [style=none] (19) at (9, -3) {};
		\node [style=none] (20) at (5, -3) {};
		\node [style=none] (21) at (7, -1) {};
		\node [style=none] (22) at (9, 1) {};
		\node [style=none] (23) at (3, 3) {};
		\node [style=none] (24) at (1, 3.5) {$(1,-1)$};
		\node [style=none] (25) at (7, 3.5) {$(1,1)$};
		\node [style=none] (26) at (1, -3.5) {$(-1,-1)$};
		\node [style=none] (27) at (7, -3.5) {$(-1,1)$};
		\node [style=none] (28) at (0, 5) {};
		\node [style=none] (29) at (0, -5) {};
		\node [style=none] (32) at (4, -4.5) {\large $h=2$};
		\node [style=none] (33) at (-6, -4.5) {\large $h=1$};
		\node [style=none] (34) at (-1, -3) {};
		\node [style=none] (35) at (-5, -3) {};
		\node [style=none] (36) at (-3, -1) {};
		\node [style=none] (37) at (-1, 1) {};
		\node [style=none] (38) at (-7, 3) {};
	\end{pgfonlayer}
	\begin{pgfonlayer}{edgelayer}
		\draw [style=wormhole] (0) to (1);
		\draw [style=wormhole] (2) to (3);
		\draw [style=axis] (6.center) to (7.center);
		\draw [style=axis] (6.center) to (8.center);
		\draw [style=wormhole] (15) to (16);
		\draw [style=wormhole] (17) to (18);
		\draw [style=particle color 1] (19.center) to (21.center);
		\draw [style=particle color 2] (21.center) to (23.center);
		\draw [style=particle color 1] (20.center) to (21.center);
		\draw [style=particle tip color 2] (21.center) to (22.center);
		\draw [style=particle 1] (28.center) to (29.center);
		\draw [style=particle color 1] (34.center) to (36.center);
		\draw [style=particle color 1] (36.center) to (38.center);
		\draw [style=particle color 2] (35.center) to (36.center);
		\draw [style=particle tip color 2] (36.center) to (37.center);
	\end{pgfonlayer}
\end{tikzpicture}

%% file: Figure-crossings.tex
\begin{tikzpicture}
	\begin{pgfonlayer}{nodelayer}
		\node [style=none] (0) at (1, -1) {};
		\node [style=none] (1) at (2, 0) {};
		\node [style=none] (2) at (3, 1) {};
		\node [style=none] (3) at (5, 3) {};
		\node [style=none] (4) at (3, -3) {};
		\node [style=none] (5) at (4, -2) {};
		\node [style=none] (6) at (5, -1) {};
		\node [style=none] (7) at (7, 1) {};
		\node [style=none] (8) at (-1, -3) {};
		\node [style=none] (9) at (-2, -2) {};
		\node [style=none] (10) at (-4, 0) {};
		\node [style=none] (11) at (0, 0) {};
		\node [style=none] (12) at (1, 1) {};
		\node [style=none] (13) at (2, 2) {};
		\node [style=none] (14) at (-2, 2) {};
		\node [style=none] (15) at (-1, 3) {};
		\node [style=none] (16) at (0, 4) {};
		\node [style=none] (17) at (-4, 4) {};
		\node [style=none] (18) at (-3, 5) {};
		\node [style=none] (19) at (-2, 6) {};
		\node [style=none] (20) at (2, 6) {};
		\node [style=none] (21) at (0, 8) {};
		\node [style=none] (22) at (7, 5) {};
		\node [style=none] (23) at (4, 8) {};
		\node [style=none] (24) at (4, 4) {};
		\node [style=none] (25) at (6, 6) {};
		\node [style=none] (26) at (4, -2) {};
		\node [style=none] (27) at (3.5, -3.5) {};
		\node [style=none] (28) at (3.5, -3.5) {$1$};
		\node [style=none] (29) at (4.5, -2.5) {$2$};
		\node [style=none] (30) at (5.5, -1.5) {$3$};
		\node [style=none] (31) at (7.5, 0.5) {$q$};
		\node [style=none] (32) at (-1.5, -3.5) {$1$};
		\node [style=none] (33) at (-2.5, -2.5) {$2$};
		\node [style=none] (34) at (-4.5, -0.5) {};
		\node [style=none] (35) at (-4.5, -0.5) {$p$};
		\node [style=small black dot] (36) at (5, 0) {};
		\node [style=small black dot] (37) at (5.5, 0.5) {};
		\node [style=small black dot] (38) at (6, 1) {};
		\node [style=small black dot] (39) at (-2, -1) {};
		\node [style=small black dot] (40) at (-2.5, -0.5) {};
		\node [style=small black dot] (41) at (-3, 0) {};
		\node [style=black dot] (44) at (2, 6) {};
		\node [style=black dot] (45) at (4, 4) {};
		\node [style=black dot] (46) at (5, 3) {};
		\node [style=black dot] (47) at (0, 4) {};
		\node [style=black dot] (48) at (2, 2) {};
		\node [style=black dot] (49) at (3, 1) {};
		\node [style=black dot] (50) at (2, 0) {};
		\node [style=black dot] (51) at (1, -1) {};
		\node [style=black dot] (52) at (1, 1) {};
		\node [style=black dot] (53) at (0, 0) {};
		\node [style=black dot] (54) at (-2, 2) {};
		\node [style=black dot] (55) at (-1, 3) {};
	\end{pgfonlayer}
	\begin{pgfonlayer}{edgelayer}
		\draw [style=axis] (8.center) to (22.center);
		\draw [style=axis] (9.center) to (25.center);
		\draw [style=axis] (10.center) to (23.center);
		\draw [style=axis] (4.center) to (17.center);
		\draw [style=axis] (26.center) to (18.center);
		\draw [style=axis] (6.center) to (19.center);
		\draw [style=axis] (7.center) to (21.center);
	\end{pgfonlayer}
\end{tikzpicture}

%% file: Figure-zones.tex
\begin{tikzpicture}
	\begin{pgfonlayer}{nodelayer}
		\node [style=singularity] (0) at (-3, 3) {};
		\node [style=singularity] (1) at (3, 3) {};
		\node [style=singularity] (2) at (-3, -3) {};
		\node [style=singularity] (3) at (3, -3) {};
		\node [style=none] (4) at (-5, -2.5) {};
		\node [style=none] (5) at (-5, -1.5) {};
		\node [style=none] (6) at (-4, -2.5) {};
		\node [style=none] (7) at (-4.25, -2.75) {$x$};
		\node [style=none] (8) at (-5.25, -1.75) {$t$};
		\node [style=none] (9) at (-3, 3.5) {$(1,-1)$};
		\node [style=none] (10) at (3, 3.5) {$(1,1)$};
		\node [style=none] (11) at (-3, -3.5) {$(-1,-1)$};
		\node [style=none] (12) at (3, -3.5) {$(-1,1)$};
		\node [style=none] (13) at (6, 0) {};
		\node [style=none] (14) at (-6, 0) {};
		\node [style=none] (15) at (0, 0) {};
		\node [style=none] (16) at (0, -1.75) {\large $III$};
		\node [style=none] (17) at (0, 1.75) {};
		\node [style=none] (18) at (3, 0) {\large $II$};
		\node [style=none] (19) at (-3, 0) {\large $I$};
		\node [style=none] (20) at (0, 1.75) {\large $III$};
	\end{pgfonlayer}
	\begin{pgfonlayer}{edgelayer}
		\draw [style=wormhole] (0) to (1);
		\draw [style=wormhole] (2) to (3);
		\draw [style=axis] (4.center) to (5.center);
		\draw [style=axis] (4.center) to (6.center);
		\draw [style=particle 2] (14.center) to (2);
		\draw [style=particle 2] (2) to (15.center);
		\draw [style=particle 2] (15.center) to (0);
		\draw [style=particle 2] (0) to (14.center);
		\draw [style=particle 2] (15.center) to (3);
		\draw [style=particle 2] (13.center) to (3);
		\draw [style=particle 2] (13.center) to (1);
		\draw [style=particle 2] (1) to (15.center);
	\end{pgfonlayer}
\end{tikzpicture}

%% file: Figure-stacked.tex
\begin{tikzpicture}
	\begin{pgfonlayer}{nodelayer}
		\node [style=singularity] (0) at (-9, 3) {};
		\node [style=singularity] (1) at (-3, 3) {};
		\node [style=singularity] (2) at (-9, -3) {};
		\node [style=singularity] (3) at (-3, -3) {};
		\node [style=none] (6) at (-11, -2.5) {};
		\node [style=none] (7) at (-11, -1.5) {};
		\node [style=none] (8) at (-10, -2.5) {};
		\node [style=none] (9) at (-10.25, -2.75) {$x$};
		\node [style=none] (10) at (-11.25, -1.75) {$t$};
		\node [style=singularity] (15) at (-9, 9) {};
		\node [style=singularity] (16) at (-3, 9) {};
		\node [style=none] (19) at (-11, 3) {};
		\node [style=none] (20) at (-7, 3) {};
		\node [style=none] (21) at (-9, 5) {};
		\node [style=none] (22) at (-11, 7) {};
		\node [style=none] (23) at (-5, 9) {};
		\node [style=none] (32) at (-1.5, 6) {\large $h=2$};
		\node [style=none] (33) at (-10.5, 0) {\large $h=1$};
		\node [style=none] (34) at (-1, -3) {};
		\node [style=none] (35) at (-5, -3) {};
		\node [style=none] (36) at (-3, -1) {};
		\node [style=none] (37) at (-1, 1) {};
		\node [style=none] (38) at (-7, 3) {};
		\node [style=none] (39) at (-1, 3.5) {};
		\node [style=none] (40) at (-1, 4.5) {};
		\node [style=none] (41) at (-2, 3.5) {};
		\node [style=none] (42) at (-1.75, 3.25) {$x$};
		\node [style=none] (43) at (-1.25, 4.25) {$t$};
	\end{pgfonlayer}
	\begin{pgfonlayer}{edgelayer}
		\draw [style=wormhole] (0) to (1);
		\draw [style=wormhole] (2) to (3);
		\draw [style=axis] (6.center) to (7.center);
		\draw [style=axis] (6.center) to (8.center);
		\draw [style=wormhole] (15) to (16);
		\draw [style=particle color 1] (19.center) to (21.center);
		\draw [style=particle color 2] (21.center) to (23.center);
		\draw [style=particle color 1] (20.center) to (21.center);
		\draw [style=particle tip color 2] (21.center) to (22.center);
		\draw [style=particle color 1] (34.center) to (36.center);
		\draw [style=particle color 1] (36.center) to (38.center);
		\draw [style=particle color 2] (35.center) to (36.center);
		\draw [style=particle tip color 2] (36.center) to (37.center);
		\draw [style=axis] (39.center) to (40.center);
		\draw [style=axis] (39.center) to (41.center);
	\end{pgfonlayer}
\end{tikzpicture}

%% file: Figure-penetrable.tex
\begin{tikzpicture}
	\begin{pgfonlayer}{nodelayer}
		\node [style=singularity] (0) at (-9, 3) {};
		\node [style=singularity] (1) at (-3, 3) {};
		\node [style=singularity] (2) at (-9, -3) {};
		\node [style=singularity] (3) at (-3, -3) {};
		\node [style=none] (4) at (-11, -2.5) {};
		\node [style=none] (5) at (-11, -1.5) {};
		\node [style=none] (6) at (-10, -2.5) {};
		\node [style=none] (7) at (-10.25, -2.75) { $x$};
		\node [style=none] (8) at (-11.25, -1.75) { $t$};
		\node [style=none] (9) at (-9, 3.5) { $(1,-1)$};
		\node [style=none] (10) at (-3, 3.5) { $(1,1)$};
		\node [style=none] (11) at (-9, -3.5) { $(-1,-1)$};
		\node [style=none] (12) at (-3, -3.5) { $(-1,1)$};
		\node [style=singularity] (13) at (1, 3) {};
		\node [style=singularity] (14) at (7, 3) {};
		\node [style=singularity] (15) at (1, -3) {};
		\node [style=singularity] (16) at (7, -3) {};
		\node [style=none] (17) at (9, -3) {};
		\node [style=none] (18) at (5, -3) {};
		\node [style=none] (19) at (7, -1) {};
		\node [style=none] (20) at (9, 1) {};
		\node [style=none] (21) at (3, 3) {};
		\node [style=none] (22) at (1, 3.5) { $(1,-1)$};
		\node [style=none] (23) at (7, 3.5) { $(1,1)$};
		\node [style=none] (24) at (1, -3.5) { $(-1,-1)$};
		\node [style=none] (25) at (7, -3.5) { $(-1,1)$};
		\node [style=none] (26) at (0, 5) {};
		\node [style=none] (27) at (0, -5) {};
		\node [style=none] (28) at (4, -4.5) {\large $h=2$};
		\node [style=none] (29) at (-6, -4.5) {\large $h=1$};
		\node [style=none] (30) at (-1, -3) {};
		\node [style=none] (31) at (-5, -3) {};
		\node [style=none] (32) at (-3, -1) {};
		\node [style=none] (33) at (-1, 1) {};
		\node [style=none] (34) at (-7, 3) {};
	\end{pgfonlayer}
	\begin{pgfonlayer}{edgelayer}
		\draw [style=wormhole] (0) to (1);
		\draw [style=wormhole] (2) to (3);
		\draw [style=axis] (4.center) to (5.center);
		\draw [style=axis] (4.center) to (6.center);
		\draw [style=wormhole] (13) to (14);
		\draw [style=wormhole] (15) to (16);
		\draw [style=particle double color 1] (17.center) to (19.center);
		\draw [style=particle double color 2] (19.center) to (21.center);
		\draw [style=particle color 1] (18.center) to (19.center);
		\draw [style=particle tip color 2] (19.center) to (20.center);
		\draw [style=particle 1] (26.center) to (27.center);
		\draw [style=particle color 1] (30.center) to (32.center);
		\draw [style=particle color 1] (32.center) to (34.center);
		\draw [style=particle double color 2] (31.center) to (32.center);
		\draw [style=particle double tip color 2] (32.center) to (33.center);
	\end{pgfonlayer}
\end{tikzpicture}

%% file: Figure-TDP-non-paradox-1.tex
\begin{tikzpicture}
	\begin{pgfonlayer}{nodelayer}
		\node [style=singularity] (0) at (-3, 3) {};
		\node [style=singularity] (1) at (3, 3) {};
		\node [style=singularity] (2) at (-3, -3) {};
		\node [style=singularity] (3) at (3, -3) {};
		\node [style=none] (4) at (-1, 3) {};
		\node [style=none] (5) at (5, -3) {};
		\node [style=none] (6) at (1, -3) {};
		\node [style=none] (8) at (1, 3) {};
		\node [style=none] (9) at (-5, -3) {};
		\node [style=none] (10) at (-1, -3) {};
		\node [style=none] (11) at (3, -1) {};
		\node [style=none] (12) at (-3, -1) {};
		\node [style=none] (13) at (5, 1) {};
		\node [style=none] (14) at (-5, 1) {};
		\node [style=none] (15) at (0, 2) {};
		\node [style=none] (16) at (-5, -1.25) {};
		\node [style=none] (17) at (-5, -0.25) {};
		\node [style=none] (18) at (-4, -1.25) {};
		\node [style=none] (19) at (-4.25, -1.5) {$x$};
		\node [style=none] (20) at (-5.25, -0.5) {$t$};
		\node [style=none] (21) at (-3, 3.5) {$(1,-1)$};
		\node [style=none] (22) at (3, 3.5) {$(1,1)$};
		\node [style=none] (23) at (-3, -3.5) {$(-1,-1)$};
		\node [style=none] (24) at (3, -3.5) {$(-1,1)$};
	\end{pgfonlayer}
	\begin{pgfonlayer}{edgelayer}
		\draw [style=wormhole] (0) to (1);
		\draw [style=wormhole] (2) to (3);
		\draw [style=particle color 1] (5.center) to (11.center);
		\draw [style=particle color 1] (11.center) to (15.center);
		\draw [style=particle color 2] (15.center) to (4.center);
		\draw [style=particle color 2] (15.center) to (8.center);
		\draw [style=particle color 2] (10.center) to (12.center);
		\draw [style=particle tip color 2] (12.center) to (14.center);
		\draw [style=particle color 1] (9.center) to (12.center);
		\draw [style=particle color 1] (12.center) to (15.center);
		\draw [style=particle color 2] (6.center) to (11.center);
		\draw [style=particle tip color 2] (11.center) to (13.center);
		\draw [style=axis] (16.center) to (17.center);
		\draw [style=axis] (16.center) to (18.center);
	\end{pgfonlayer}
\end{tikzpicture}

%% file: Figure-TDP-non-paradox-2.tex
\begin{tikzpicture}
	\begin{pgfonlayer}{nodelayer}
		\node [style=singularity] (0) at (-3, 3) {};
		\node [style=singularity] (1) at (3, 3) {};
		\node [style=singularity] (2) at (-3, -3) {};
		\node [style=singularity] (3) at (3, -3) {};
		\node [style=none] (4) at (-1, 3) {};
		\node [style=none] (5) at (5, -3) {};
		\node [style=none] (6) at (1, -3) {};
		\node [style=none] (8) at (1, 3) {};
		\node [style=none] (9) at (-5, -3) {};
		\node [style=none] (10) at (-1, -3) {};
		\node [style=none] (11) at (3, -1) {};
		\node [style=none] (12) at (-3, -1) {};
		\node [style=none] (13) at (5, 1) {};
		\node [style=none] (14) at (-5, 1) {};
		\node [style=none] (15) at (0, 2) {};
		\node [style=none] (16) at (-5, -1.25) {};
		\node [style=none] (17) at (-5, -0.25) {};
		\node [style=none] (18) at (-4, -1.25) {};
		\node [style=none] (19) at (-4.25, -1.5) {$x$};
		\node [style=none] (20) at (-5.25, -0.5) {$t$};
		\node [style=none] (21) at (-3, 3.5) {$(1,-1)$};
		\node [style=none] (22) at (3, 3.5) {$(1,1)$};
		\node [style=none] (23) at (-3, -3.5) {$(-1,-1)$};
		\node [style=none] (24) at (3, -3.5) {$(-1,1)$};
	\end{pgfonlayer}
	\begin{pgfonlayer}{edgelayer}
		\draw [style=wormhole] (0) to (1);
		\draw [style=wormhole] (2) to (3);
		\draw [style=particle color 1] (5.center) to (11.center);
		\draw [style=particle color 2] (11.center) to (15.center);
		\draw [style=particle color 1] (15.center) to (4.center);
		\draw [style=particle color 1] (15.center) to (8.center);
		\draw [style=particle color 1] (10.center) to (12.center);
		\draw [style=particle tip color 2] (12.center) to (14.center);
		\draw [style=particle color 1] (9.center) to (12.center);
		\draw [style=particle color 2] (12.center) to (15.center);
		\draw [style=particle color 1] (6.center) to (11.center);
		\draw [style=particle tip color 2] (11.center) to (13.center);
		\draw [style=axis] (16.center) to (17.center);
		\draw [style=axis] (16.center) to (18.center);
	\end{pgfonlayer}
\end{tikzpicture}

%% file: TimeTravelParadoxes.bbl
\providecommand{\href}[2]{#2}\begingroup\raggedright\begin{thebibliography}{10}

\bibitem{FTL_TT}
B.~Shoshany, ``{Lectures on Faster-than-Light Travel and Time Travel},''
  \href{http://dx.doi.org/10.21468/SciPostPhysLectNotes.10}{{\em SciPost Phys.
  Lect. Notes 10} (2019) }.

\bibitem{Visser}
M.~Visser, {\em {Lorentzian Wormholes}}.
\newblock American Inst. of Physics, 1996.
\newblock \url{https://www.springer.com/gp/book/9781563966538}.

\bibitem{Krasnikov}
S.~Krasnikov, \href{http://dx.doi.org/10.1007/978-3-319-72754-7}{{\em
  {Back-in-Time and Faster-than-Light Travel in General Relativity}}}.
\newblock Springer International Publishing, 2018.

\bibitem{Lobo}
F.~S.~N. Lobo, ed., \href{http://dx.doi.org/10.1007/978-3-319-55182-1}{{\em
  {Wormholes, Warp Drives and Energy Conditions}}}.
\newblock Springer International Publishing, 2017.

\bibitem{EinsteinRosen35}
A.~Einstein and N.~Rosen, ``The particle problem in the general theory of
  relativity,'' \href{http://dx.doi.org/10.1103/PhysRev.48.73}{{\em Phys. Rev.}
  {\bfseries 48} (Jul, 1935) 73--77}.

\bibitem{VanStockum38}
W.~J. van Stockum, ``Ix.-the gravitational field of a distribution of particles
  rotating about an axis of symmetry,''
  \href{http://dx.doi.org/10.1017/S0370164600013699}{{\em Proceedings of the
  Royal Society of Edinburgh} {\bfseries 57} (Aug, 1938) 135--154}.

\bibitem{Godel49}
K.~G\"odel, ``An example of a new type of cosmological solutions of einstein's
  field equations of gravitation,''
  \href{http://dx.doi.org/10.1103/RevModPhys.21.447}{{\em Rev. Mod. Phys.}
  {\bfseries 21} (Jul, 1949) 447--450}.

\bibitem{Alcubierre94}
M.~Alcubierre, ``The warp drive: hyper-fast travel within general relativity,''
  \href{http://dx.doi.org/10.1088/0264-9381/11/5/001}{{\em Classical and
  Quantum Gravity} {\bfseries 11} no.~5, (May, 1994) L73--L77},
  \href{http://arxiv.org/abs/gr-qc/0009013}{{\ttfamily arXiv:gr-qc/0009013}}.

\bibitem{MorrisThorne88}
M.~S. Morris and K.~S. Thorne, ``{Wormholes in spacetime and their use for
  interstellar travel: A tool for teaching general relativity},''
  \href{http://dx.doi.org/10.1119/1.15620}{{\em American Journal of Physics}
  {\bfseries 56} no.~5, (May, 1988) 395--412}.
  \url{http://www.cmp.caltech.edu/refael/league/thorne-morris.pdf}.

\bibitem{FrolovNovikov90}
V.~P. Frolov and I.~D. Novikov, ``Physical effects in wormholes and time
  machines,'' \href{http://dx.doi.org/10.1103/PhysRevD.42.1057}{{\em Phys. Rev.
  D} {\bfseries 42} (Aug, 1990) 1057--1065}.

\bibitem{Gott91}
J.~R. Gott, III, ``{Closed timelike curves produced by pairs of moving cosmic
  strings: Exact solutions},''
\href{http://dx.doi.org/10.1103/PhysRevLett.66.1126}{{\em Phys. Rev. Lett.}
  {\bfseries 66} (1991) 1126--1129}.

\bibitem{Krasnikov98}
S.~V. Krasnikov, ``Hyperfast travel in general relativity,''
  \href{http://dx.doi.org/10.1103/PhysRevD.57.4760}{{\em Phys. Rev. D}
  {\bfseries 57} (Apr, 1998) 4760--4766}.

\bibitem{EverettRoman97}
A.~E. Everett and T.~A. Roman, ``Superluminal subway: The krasnikov tube,''
  \href{http://dx.doi.org/10.1103/PhysRevD.56.2100}{{\em Phys. Rev. D}
  {\bfseries 56} (Aug, 1997) 2100--2108}.

\bibitem{Curiel:2014zba}
E.~Curiel, \href{http://dx.doi.org/10.1007/978-1-4939-3210-8_3}{``{A Primer on
  Energy Conditions},''} in {\em Towards a Theory of Spacetime Theories},
  D.~Lehmkuhl, G.~Schiemann, and E.~Scholz, eds., pp.~43--104.
\newblock Springer New York, New York, NY, 2017.
\newblock \href{http://arxiv.org/abs/1405.0403}{{\ttfamily arXiv:1405.0403}}.

\bibitem{Krasnikov02}
S.~Krasnikov, ``{The Time travel paradox},''
  \href{http://dx.doi.org/10.1103/PhysRevD.65.064013}{{\em Phys. Rev.}
  {\bfseries D65} (2002) 064013},
\href{http://arxiv.org/abs/gr-qc/0109029}{{\ttfamily arXiv:gr-qc/0109029}}.

\bibitem{hawking_ellis_1973}
S.~W. Hawking and G.~F.~R. Ellis,
  \href{http://dx.doi.org/10.1017/CBO9780511524646}{{\em {The Large Scale
  Structure of Space-Time}}}.
\newblock Cambridge Monographs on Mathematical Physics. Cambridge University
  Press, 1973.

\bibitem{Hawking92}
S.~W. Hawking, ``Chronology protection conjecture,''
  \href{http://dx.doi.org/10.1103/PhysRevD.46.603}{{\em Phys. Rev. D}
  {\bfseries 46} (Jul, 1992) 603--611}.

\bibitem{Deutsch91}
D.~Deutsch, ``Quantum mechanics near closed timelike lines,''
  \href{http://dx.doi.org/10.1103/PhysRevD.44.3197}{{\em Phys. Rev. D}
  {\bfseries 44} (Nov, 1991) 3197--3217}.

\bibitem{Novikov90}
J.~Friedman, M.~S. Morris, I.~D. Novikov, F.~Echeverria, G.~Klinkhammer, K.~S.
  Thorne, and U.~Yurtsever, ``Cauchy problem in spacetimes with closed timelike
  curves,'' \href{http://dx.doi.org/10.1103/PhysRevD.42.1915}{{\em Phys. Rev.
  D} {\bfseries 42} (Sep, 1990) 1915--1930}.
  \url{https://authors.library.caltech.edu/3737/}.

\bibitem{Consortium91}
F.~Echeverria, G.~Klinkhammer, and K.~S. Thorne, ``{Billiard balls in wormhole
  spacetimes with closed timelike curves: Classical theory},''
  \href{http://dx.doi.org/10.1103/PhysRevD.44.1077}{{\em Phys. Rev. D}
  {\bfseries 44} (Aug, 1991) 1077--1099}.
  \url{https://authors.library.caltech.edu/6469/}.

\bibitem{Friedman97}
J.~L. Friedman and M.~S. Morris, ``{Existence and uniqueness theorems for
  massless fields on a class of space-times with closed timelike curves},''
  \href{http://dx.doi.org/10.1007/s002200050118}{{\em Commun. Math. Phys.}
  {\bfseries 186} (1997) 495--530},
\href{http://arxiv.org/abs/gr-qc/9411033}{{\ttfamily arXiv:gr-qc/9411033
  [gr-qc]}}.

\bibitem{McCabe}
G.~McCabe, ``{The Topology of Branching Universes},''
  \href{http://dx.doi.org/10.1007/s10702-005-1319-9}{{\em Foundations of
  Physics Letters} {\bfseries 18} no.~7, (Nov, 2005) 665--676},
  \href{http://arxiv.org/abs/gr-qc/0505150}{{\ttfamily arXiv:gr-qc/0505150}}.

\bibitem{Penrose79}
R.~Penrose, ``{Singularities and time-asymmetry},'' in {\em {General
  Relativity: An Einstein Centenary Survey}}, pp.~581--638.
\newblock
1979.
\newblock

\bibitem{Visser93}
M.~Visser, ``From wormhole to time machine: Remarks on hawking's chronology
  protection conjecture,''
  \href{http://dx.doi.org/10.1103/PhysRevD.47.554}{{\em Physical review D:
  Particles and fields} {\bfseries 47} (02, 1993) 554--565}.

\bibitem{MikheevaNovikov93}
E.~V. Mikheeva and I.~D. Novikov, ``{Inelastic billiard ball in a space-time
  with a time machine},''
\href{http://dx.doi.org/10.1103/PhysRevD.47.1432}{{\em Phys. Rev.} {\bfseries
  D47} (1993) 1432--1436}.

\bibitem{Politzer92}
H.~D. Politzer, ``Simple quantum systems in spacetimes with closed timelike
  curves,'' \href{http://dx.doi.org/10.1103/PhysRevD.46.4470}{{\em Phys. Rev.
  D} {\bfseries 46} (Nov, 1992) 4470--4476}.

\bibitem{KalyanaRama:1994ag}
S.~Kalyana~Rama and S.~Sen, ``{Inconsistent physics in the presence of time
  machines},''
\href{http://arxiv.org/abs/gr-qc/9410031}{{\ttfamily arXiv:gr-qc/9410031}}.

\bibitem{Krasnikov97}
S.~Krasnikov, ``Causality violation and paradoxes,''
  \href{http://dx.doi.org/10.1103/PhysRevD.55.3427}{{\em Physical Review D}
  {\bfseries 55} no.~6, (1997) 3427}.

\bibitem{munkres2000topology}
J.~Munkres, {\em Topology}.
\newblock Featured Titles for Topology. Prentice Hall, Incorporated, 2000.
\newblock \url{https://books.google.com/books?id=XjoZAQAAIAAJ}.

\bibitem{Hatcher:478079}
A.~Hatcher, {\em {Algebraic topology}}.
\newblock Cambridge Univ. Press, Cambridge, 2000.
\newblock \url{https://cds.cern.ch/record/478079}.

\bibitem{Norton}
J.~D. Norton, ``Causation as folk science,'' {\em Philosophers' Imprint}
  {\bfseries 3} no.~4, (2003) 1--22.
  \url{http://hdl.handle.net/2027/spo.3521354.0003.004}.

\end{thebibliography}\endgroup
